\pdfoutput=1
\documentclass[
reprint,
showkeys,
aps,
prx, twocolumn,
floatfix,
]{revtex4-2}

\usepackage{amsmath,amssymb}
\usepackage{graphicx}
\usepackage{bm}
\usepackage{xcolor}

\usepackage{colortbl}

\usepackage[utf8]{inputenc}
\usepackage[cal=boondoxo]{mathalfa}
\usepackage[ruled, linesnumbered]{algorithm2e}
\usepackage[english]{babel}
\usepackage{mathtools}
\usepackage{algorithm2e}
\usepackage[textsize=footnotesize]{todonotes}  

\usepackage{xpatch}
\makeatletter
\xpatchcmd{\@ssect@ltx}{\@xsect}{\protected@edef\@currentlabelname{#8}\@xsect}{}{}
\xpatchcmd{\@sect@ltx}{\@xsect}{\protected@edef\@currentlabelname{#8}\@xsect}{}{}
\makeatother

\usepackage[colorlinks,linkcolor=darkblue,citecolor=darkblue,linktocpage,hypertexnames=true,pdfpagelabels]{hyperref}
\usepackage[capitalise]{cleveref}
\usepackage{subcaption}
\usepackage{hyphenat}
\usepackage{verbatim}

\usepackage{amsthm}
\newtheorem{theorem}{Theorem}
\newtheorem{corollary}{Corollary}

\newtheorem{definition}{Definition}
\newtheorem{proposition}{Proposition}

\newtheorem{conjecture}{Conjecture}

\usepackage{tikz}
\usetikzlibrary{shapes.geometric, calc, math, arrows, decorations}
\usepackage{tkz-euclide}
\usetikzlibrary{backgrounds}

\newcommand{\mc}{\mathcal}
\newcommand{\mr}{\mathrm}

\usepackage{enumitem}

\definecolor{darkblue}{rgb}{0.0, 0.0, 0.55}
\definecolor{amethyst}{rgb}{0.6, 0.4, 0.8}
\definecolor{ao}{rgb}{0.0, 0.0, 1.0}
\definecolor{amber}{rgb}{1.0, 0.75, 0.0}
\definecolor{amaranth}{rgb}{0.9, 0.17, 0.31}

\begin{document}

\title{The grammar of the Ising model: A new complexity hierarchy}

\author{Tobias Reinhart}
\email{tobias.reinhart@uibk.ac.at}
\author{Gemma De les Coves}
\affiliation{Institute for Theoretical Physics, University of Innsbruck, Technikerstr.\ 21a,\ A-6020 Innsbruck, Austria}

\begin{abstract}
How complex is an Ising model? Usually, this is measured by the computational complexity of its ground state energy problem. Yet, this complexity measure only distinguishes between planar and non-planar interaction graphs, and thus fails to capture properties such as the average node degree, the number of long range interactions, or the dimensionality of the lattice. Herein, we introduce a new complexity measure for Ising models and thoroughly classify Ising models with respect to it. Specifically, given an Ising model we consider the decision problem corresponding to the function graph of its Hamiltonian, and classify this problem in the Chomsky hierarchy. We prove that the language of this decision problem is (i) regular if and only if the Ising model is finite, (ii) constructive context free if and only if the Ising model is linear and its edge language is regular, (iii) constructive context sensitive if and only if the edge language of the Ising model is context sensitive, and (iv) decidable if and only if the edge language of the Ising model is decidable. We apply this theorem to show that the 1d Ising model, the Ising model on generalised ladder graphs, and the Ising model on layerwise complete graphs are constructive context free, while the 2d Ising model, the all-to-all Ising model, and the Ising model on perfect binary trees are constructive context sensitive. 
This work is a first step in the characterisation of physical interactions in terms of grammars.
\end{abstract}


\maketitle

\section{Introduction}
\label{sec:intro}

\begin{figure*}[t]\centering
    \includegraphics[width=.9\textwidth]{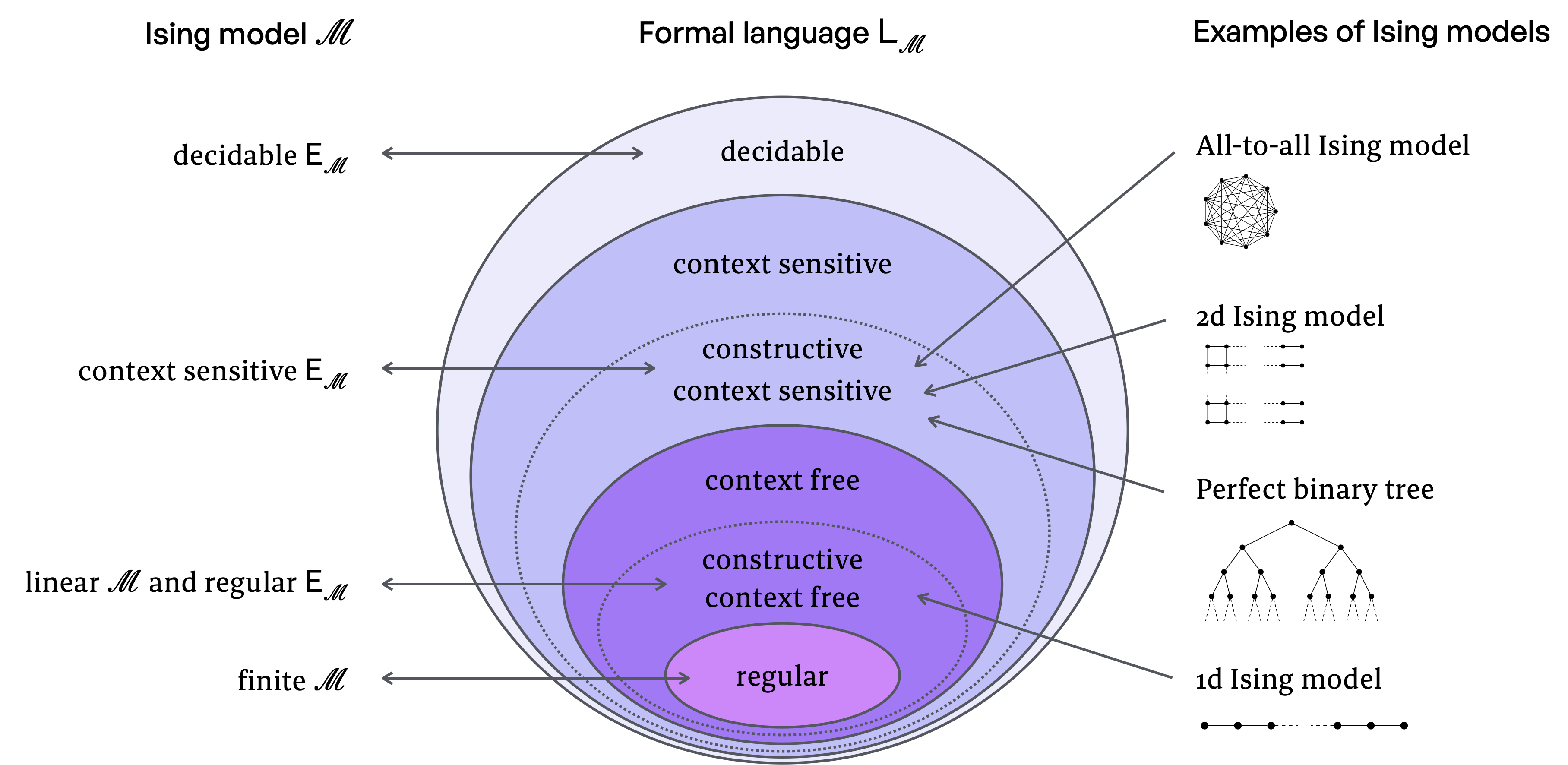}
    \caption{(Left and Middle)
We fully classify the properties of an Ising model $\mc{M}$ that determines the complexity  of its language ${\sf L}_{\mc{M}}$ in the refined Chomsky hierarchy  (\cref{thm:main}).  
    (Right) 
We apply this classification to show   that the language of the 
    1d Ising model is constructive context free;
    the language of the Ising model on perfect binary trees, 
    2d lattices, 
    and all-to-all interaction graphs is constructive context sensitive (\cref{sec:examples}). 
    }
    \label{fig:mainresult}
\end{figure*}

Spin models are a powerful tool to model complex systems.
While the paradigmatic spin model, the Ising model \cite{Is25,Pa96,Mc13}, was originally proposed as a stripped-off model of magnetism, it has since been used in a remarkable variety of settings, including 
as a toy model of matter in certain quantum gravity models \cite{Am09a}, 
to model gases (via so-called lattice gas models) \cite{Ch87}, 
in knot theory (via the connection of the Jones polynomial with the partition function of the Potts model in a certain parameter regime) \cite{Ka01}, 
for artificial neural networks  (stemming from Hopfield's proposal) \cite{Ho82,Am85}, 
in ecology (e.g.\ to model the size of canopy trees) \cite{So00}, 
to model flocks of birds \cite{Bi12}, 
viruses as quasi-species \cite{An83,Ta92}, 
genetic interactions \cite{Le06}, 
for protein folding \cite{Ba03,Ek13,Le86,Le87} (together with its generalisation, the Potts model \cite{Ri20}),  
for economic opinions, urban segregation and language change  \cite{St08}, 
for random language models \cite{De19f}, 
social dynamics \cite{Ca09}, 
earthquakes \cite{Ji07}  
and the US Supreme Court \cite{Le15c}, to name some. 
The relevant questions differ for each of these applications---e.g.\ for artificial neural networks, one is interested in a ``driven'' Ising model, where the parameters are updated (corresponding to learning) and one may study convergence rates, 
whereas in complex systems \cite{So00}, one may be interested in the behaviour of the Ising model at criticality. 
Whatever the focus may be, the fact is that this very simplified model provides insights into very different problems. 

Depending on the application, the Ising model is considered on different families of interaction graphs, 
such as lattices of a certain dimensionality for magnetism, 
or (layerwise) complete graphs for artificial neural networks. 
How does one characterise these different Ising models?
In particular, how can we measure the complexity of  Ising models on different families of graphs? 
Traditionally, this is measured by the computational complexity of the ground state energy problem (GSE), which asks:
\begin{quote}
Given an interaction graph for $n$ spins and an integer $k$, 
does there exist a spin configuration with energy below $k$?  
\end{quote}
For the Ising model without fields, if the family of interaction graphs is planar, this problem is in P, 
and if it is non-planar, it is NP-complete \cite{Bar82,Is00}\footnote{P (NP) is the class of decision problems solvable in polynomial time by a (non)-deterministic Turing machine.}.  
These results have given rise to strong and fruitful ties between spin models (and, more generally, statistical mechanics) and computational complexity \cite{Mo11,Me09}. 
For example, one can formulate many NP problems in terms of the GSE \cite{Lu14b}.

Yet, this measure is very coarse: 
It only classifies Ising models depending on whether they are defined on planar or non-planar graphs (resulting in a two-level hierarchy of P and NP-complete, respectively). 
It is insensitive to the dimensionality of the interaction graph (when considering lattices), 
the number of long range interactions, 
or the average node degree. 
This might be due to the facts that GSE only `cares' about the low energy sector of the model, 
and that computational complexity tends to gloss over polynomial factors. 
Clearly, a 1d Ising model has a different local structure than a 2d Ising model, yet this distinction is invisible in the traditional measure. 
Can one devise a measure that captures the complexity of the local structure of an Ising model?

In this work, 
we introduce a new complexity measure for Ising models 
and thoroughly classify them in it. 
We do so in three steps. 
First, given an Ising model $\mc{M}$ we define its language ${\sf L}_{\mc{M}}$ which encodes the function graph of its Hamiltonian $H_{\mc{M}}$, that is, the set of all pairs of spin configurations and their energy, 
\begin{equation}
{\sf L}_{\mc{M}} = \{(x,H_{\mc{M}}(x)) \mid x \textrm{ is a spin configuration of } \mc{M}\} 
\end{equation}
Second, we consider the problem of deciding ${\sf L}_{\mc{M}}$, that is: 
\begin{quote}
Given $(x,E)$, is $x$ is a valid spin configuration of $\mc{M}$ and is $E = H_{\mc{M}}(x)$?
\end{quote}
Third, we measure the hardness of this problem by classifying ${\sf L}_{\mc{M}}$ in a (refined) Chomsky hierarchy.
This answers the question: 
\begin{quote}
What is the simplest type of grammar (automaton) that generates (accepts) ${\sf L}_{\mc{M}}$? 
\end{quote}
In this sense, our complexity measure captures the minimal complexity of local rules (of the grammar or automaton) that suffices to reproduce the Hamiltonian of an Ising model.

Note that in contrast to this, computational complexity does not capture the complexity of local rules but rather the resources (in time, space or non-determinism) that a Turing machine needs to recognise the language. 
Since the language of various Ising models can be decided in polynomial time by a deterministic Turing machine, 
computational complexity is insufficient to distinguish among their local structures, 
whereas the Chomsky hierarchy can achieve that.  

The refinement of the Chomsky hierarchy stems from posing restrictions on the automata that accept context free and context sensitive languages, 
resulting in two levels called constructive context free and constructive context sensitive (\cref{fig:mainresult}), 
which we conjecture to coincide with context free and context sensitive, respectively.  

We prove that (\cref{thm:main}): 
\begin{enumerate}[label=(\roman*)]
    \item ${\sf L}_{\mc{M}}$ is regular if and only if $\mc{M}$ is finite;
    \item ${\sf L}_{\mc{M}}$ is constructive context free if and only if 
    $\mc{M}$ is linear and ${\sf E}_{\mc{M}}$ is regular;
    \item ${\sf L}_{\mc{M}}$ is constructive context sensitive if and only if 
    ${\sf E}_{\mc{M}}$ is context sensitive; and 
    \item ${\sf L}_{\mc{M}}$ is  decidable if and only if 
    ${\sf E}_{\mc{M}}$ is decidable. 
\end{enumerate}
This classification fully characterises the complexity of ${\sf L}_{\mc{M}}$ in terms of 
properties of the interaction graphs of  $\mc{M}$. 
Specifically, the edge language ${\sf E}_{\mc{M}}$ encodes which spins interact, 
and its complexity captures how difficult it is to decide whether two spins interact or not. 
The remaining properties of $\mc{M}$ (being finite or linear) encode how the number of interactions grows with the system size. 

We then apply this classification to common families of interaction graphs, and show that  (\cref{fig:mainresult2}): 
\begin{enumerate}[label=(\roman*)]
\item  The language of the 
1d Ising model with open or periodic boundary conditions, 
the Ising model on ladder graphs, 
and the Ising model on layerwise complete graphs 
is constructive context free. 
All of these Ising models are linear, and their edge language is regular. 

\item The language of the 
Ising model on perfect binary trees
and the 2d Ising model 
is constructive context sensitive. 
All of these Ising models are linear, and their edge language is context sensitive.

\item 
The language of the all-to-all Ising model 
is constructive context sensitive.  
This Ising model is not linear and its edge language is regular. 

\end{enumerate}

Note that a similar approach was recently proposed in \cite{St21}, yet for a more general definition of spin model which only achieves a partial characterisation in the Chomsky hierarchy. 
The present focus on the Ising model allows us to promote the partial characterisation to a thorough classification. 
In particular, 
this work identifies which properties of the interaction graphs play a role in the complexity of the model, 
and specifies how they interact (metaphorically) to increase the complexity.  

This paper is structured as follows. 
In \cref{sec:stage} we define the new complexity measure for Ising models. 
In \cref{sec:mainresult} we define several properties of Ising models and state our main result, \cref{thm:main}, which we prove in \cref{sec:proofmain}.
In \cref{sec:examples} we apply our main result to obtain the complexity of some well-known examples of Ising models, and in \cref{sec:conclusions} we conclude and present an outlook. 
In \cref{ssec:FormalLanguage} we list some basic definitions and results from formal language and automata theory.

\section{A new complexity measure for Ising models}
\label{sec:stage}
Here we define the new complexity measure for Ising models. 
First, in \cref{ssec:formal-lan} we provide an intuitive introduction to formal language theory.
In \cref{ssec:Isingmodel} we explain how the function graph of the Hamiltonian of an Ising model $\mc{M}$ is encoded as a formal language ${\sf L}_{\mc{M}}$.
In \cref{ssec:Chomsky} we define how the complexity measure for Ising models is obtained by classifying ${\sf L}_{\mc{M}}$ in the Chomsky hierarchy. 

\subsection{Formal language theory}\label{ssec:formal-lan}
Formal language theory concerns the study of \emph{formal languages}, that is sets of strings over a finite alphabet $\Sigma$. Denote by $\Sigma^*$ the set of all finite strings over $\Sigma$, including the \emph{empty string} $\epsilon$ then a formal language is a subset $\textsf{L}\subseteq \Sigma^*$.

One of the main goals of formal language theory is to classify formal languages w.r.t.\ their complexity, e.g.\ by considering the simplest \emph{grammar} (\cref{def:grammar}) that can produce a given formal language $\textsf{L}$.
Formal grammars consist of a finite set of \emph{production rules} that transform one string into another.
Each grammar further divides the symbols appearing on the right- and left-hand sides of its production rules into a finite set of \emph{terminal} and a finite set of \emph{non-terminal} symbols. 
The language that a given grammar \emph{generates} then consists of all strings that only contain terminal symbols and that can be obtained by applying any finite sequence of production rules to a distinguished \emph{start symbol}.

Consider for instance the alphabet $\{a,b\}$ and the formal language $\textsf{L}_{a^nb^n} = \{a^nb^n \mid n \geq 1 \}$, where $a^n$ denotes the string consisting of $n$-repetitions of the symbol $a$. This language is generated by the grammar with terminal symbols $\{a,b\}$, non-terminal symbols $\{S\}$, start symbol $S$ and production rules 
\begin{equation}
    \begin{split}
        S &\to aSb \\
        S &\to ab.
    \end{split}
\end{equation}

Depending on the strings appearing in the production rules, grammars can be divided into different types. 
The \emph{Chomsky hierarchy} (\cite{Ch63, Ch65}, see \cref{def:Chomsky}) is based on the grammar types \emph{regular}, \emph{context free}, \emph{context sensitive} and \emph{unrestricted}.
These grammar types form an inclusion hierarchy, i.e.\ each regular grammar is also context free, each context free grammar is also context sensitive, etc. 
The Chomsky hierarchy can be extended to languages, by calling a formal language regular if it can be generated by a regular grammar and similar for the remaining levels. 
In this sense the Chomsky hierarchy defines a complexity hierarchy for formal languages.  

For each level of the Chomsky hierarchy there exists a corresponding type of \emph{automaton}. 
Automata are abstract computing machines that consist of
\begin{enumerate}
    \item an (infinite) \emph{input tape} which is divided into cells, each of which can store a single symbol
    \item a finite number of \emph{internal states}
    \item possibly additional \emph{external memory} 
    \item a \emph{machine head} that moves over the input tape 
\end{enumerate}
In each step of the computation, an automaton scans one symbol from its input tape and then, depending on the scanned symbol, the internal state and possibly the content of the external memory, performs an action that can consist of 
moving its head one cell to the left/right, overwriting the current cell on the tape, storing a finite string of symbols in the additional memory and changing its internal state.
How precisely the current symbol, state and memory content give rise to the next action is defined by the \emph{transition relation} of the automaton. 
This transition relation is hardwired, that is each automaton has a pre-defined transition relation.

An automaton is said to \emph{accept} an input string if with the string written to its input tape after a finite number of steps its computation ends in a distinguished accept state.
Each automaton defines a formal language, the language of strings that it accepts. 
In this sense, automata, similar to grammars provide \emph{finite descriptions} of formal languages.

Similar to grammars, depending on the form of their transition relation, automata can be divided into different types. To each type of grammar of the Chomsky hierarchy there exists a corresponding type of automaton such that a formal language 
can be generated by the respective type of grammar if and only if it is accepted by the corresponding type of automaton. 

The automata corresponding to regular grammars are \emph{finite state automata} (FSAs) (\cref{def:FSA}). FSAs have no external memory, their tape is read-only, and they can only move their head to the right. 

The automata corresponding to context free grammars are \emph{push down automata} (PDAs) (\cref{def:PDA}). 
PDAs are FSAs that have an additional, external stack memory, that is a last-in-first-out memory. 

The automata corresponding to unrestricted grammars are \emph{Turing machines} (TMs) (\cref{def:TM}). TMs have an infinite tape to which they can also write. Moreover, they can move their head left or right.

The automata corresponding to context sensitive grammars are \emph{linear bounded automata} (LBAs) (\cref{def:LBA}). LBAs are Turing machines that, during their entire computation, have their input tape restricted to those cells which are originally occupied by the input string.

Constructing an automaton of a certain type that accepts a given language proves that this language belongs to the corresponding level of the Chomsky hierarchy and thereby provides an upper bound for its complexity. 

For instance, the language $\textsf{L}_{a^nb^n}$ can easily be proven to be context free, by constructing a PDA that accepts it. On input $a^nb^n$ the PDA first, for each $a$ that it scans pushes one $a$ to its stack, then, for each $b$ that it scans, it pops one $a$ from its stack. Finally, once it has scanned the entire input, it accepts if and only if its stack is empty, as then the number of $a$s and $b$s is equal.    

Proving lower bounds for the complexity of a language amounts to proving that it cannot be generated/accepted by a grammar/an automaton of a certain type. There exist several theorems and techniques that can be used to achieve this, including closure properties of formal languages and pumping lemmas. Details can be found in \cite{Ko97}.

\subsection{The Ising model and its language}
\label{ssec:Isingmodel}

What is an Ising model? 
That depends on the context: 
the system size may be pre-determined, unspecified, or defined in the thermodynamic limit. 
Additionally, the couplings may be fixed or drawn from a probability distribution, in which case it is usually called a spin glass. In this work, an Ising model is defined as follows. 

\begin{definition}[Ising model]
    \label{def:IsingModel}
    An \emph{Ising model} $\mc{M}$ is a pair $\mc{M}= (N_{\mc{M}}, E_{\mc{M}})$, 
    where  
   \begin{eqnarray}
   &&N_{\mc{M}}\subseteq \mathbb{N} \\
   &&E_\mc{M}   = \{ (E_\mc{M} )_n \mid n \in N_{\mc{M}}\}
   \end{eqnarray}
    and 
       \begin{equation}
    (E_\mc{M} )_n \subseteq \{(i,j) \mid i,j \in \{1,\ldots, n\}, i < j\}
       \end{equation}
    defines an undirected, ordered graph with vertex set $V_n \coloneqq \{1,\ldots, n\}$ that has no isolated vertices. 
    An Ising model $\mc{M}$ defines a \emph{Hamiltonian} $H_{\mc{M}}$: 
    \begin{align}
    \begin{aligned}\label{eqn:IsingHamiltonian}
    H_{\mc{M}} : \bigcup _{n \in N_{\mc{M}}} \{0,1\}^n & \rightarrow \mathbb{Z} \\
    s_1\ldots s_n & \mapsto \sum_{(i,j) \in (E_{\mc{M}})_n} h(s_i,s_j)
    \end{aligned}
    \end{align}
    where $h(s_i,s_j)=-1$ if $s_i=s_j$ and $+1$ else. 
\end{definition}

In words, $N_{\mc{M}}$ specifies the system sizes for which $\mc{M}$ is defined, 
and for each $n \in N_{\mc{M}}$, 
the edge set $(E_{\mc{M}})_n$ describes how the system of $n$ spins interacts. 
Specifically, if there is an edge $(i,j)\in (E_{\mc{M}})_n$, \emph{spins} $i$ and $j$ interact. 
We require that no vertex is isolated (i.e.\ that every vertex is contained in at least one edge) 
because isolated vertices correspond to non-interacting spins, 
which do not contribute to the Hamiltonian.  

A \emph{configuration} of spins assigns a \emph{state} from $\{0,1\}$ to each of the spins. Since we assume that the spins are enumerated, configurations of $n$ spins correspond to strings $s_1\ldots s_n$ of $n$-numbers from $\{0,1\}$.
Note that we take states to be from $\{0,1\}$ instead of the (more) common choice $\{-1,+1\}$. Nevertheless,  
$h(s_i,s_j)$ defines an Ising interaction, i.e.\ only depends on the parity of the two spins $i$ and $j$.    
The \emph{Hamiltonian} $H_{\mc{M}}$ of an Ising model $\mc{M}$ then maps configurations $s_1 \ldots s_n$ to the sum of their local energies.
Note that $H_{\mc{M}}$ is defined for configurations of all system sizes from $N_{\mc{M}}$, i.e.\ could also be interpreted as a family of functions, one for each system size. 

\cref{def:IsingModel} could be generalised to include non-constant couplings or higher order interactions, 
e.g.\ by using hyperedge-labeled hypergraphs (see e.g.\ \cite{De13b}), as done in \cite{St21}.  
Yet, in this work we focus on the constant coupling case in order to classify the complexity of Ising models solely based on their interaction structure. 

Note that $\mc{M}$ is generally defined for an infinite set of system sizes, 
and that it is not defined in the thermodynamic limit ($n\to \infty$). 
Both are crucial for encoding  $H_{\mc{M}} $ as a language, 
as finitely many system sizes would result in a finite language (which is trivially regular), 
and the thermodynamic limit would require infinite strings (precluding the use of formal languages). 

Finally, note that in \cref{def:IsingModel} 
the vertices of the interaction graphs have an order, 
as imposing such an order is necessary when encoding graphs as strings, and thus a family of graphs as a language. 
Ultimately this is due to the fact that symbols in a string have a canonical order, 
while vertices in a graph do not.
Disposing of the order of the vertices 
would require considering equivalence classes of encodings of graphs (where two strings are equivalent if they encode the same graph), 
and measuring the complexity of an Ising model would require a minimisation over all equivalent encodings of that Ising model. 
The latter would involve, in particular, solving the graph isomorphism problem. 
Alternatively, such an order could be disposed of by casting spin models as graph languages (see the \nameref{sec:conclusions}).

In order to define the language of an Ising model, let $u$ denote the unary encoding of integers
\begin{align}\label{eqn:Unary}
    \begin{aligned}
        u : \mathbb{Z} & \rightarrow \{+,-\}^*\\
        u(z) &\coloneqq 
        \begin{cases}
                 \epsilon & \text{if} \ z=0 \\
                 +^z & \text{if} \ z>0 \\
                 -^{-z} & \text{else}
                 \end{cases}
    \end{aligned}
\end{align}
Note that here $+$ and $-$ are just symbols, not mathematical operations.
 
\begin{definition}[Language of an Ising model]
\label{def:HamLanguage}
Let $\mc{M}$ be an Ising model. 
The \emph{language of}  $\mc{M}$, ${\sf L}_{\mc{M}}$, is defined as  
\begin{multline}\label{eqn:HamLanguage}
    {\sf L}_{\mc{M}} \coloneqq \{ s_1\ldots s_n \bullet u( H_{\mc{M}}(s_1 \ldots s_n) )\mid \\
    n \in N_{\mc{M}}, s_i \in \{0,1\} \}  
\end{multline}
\end{definition}

In words, ${\sf L}_{\mc{M}}$ encodes the function graph of $H_{\mc{M}}$. 
Explicitly,
we use the symbol $\bullet$ as a separator between spin configurations and energies.  
Let $\sigma \in \{+,-\}$.  
A string $s_1\ldots s_n\bullet \sigma^k$
is contained in ${\sf L}_{\mc{M}}$ if 
$s_1\ldots s_n$ is a spin configuration from the domain of $H_{\mc{M}}$
and $\sigma^k$ equals the unary encoding of $H_{\mc{M}}(s_1\ldots s_n)$. 

Note that the energy is encoded in unary,   
as this leads to a more fine-grained classification in the Chomsky hierarchy. 
Specifically,  
encoding the energy in binary would render the addition of individual energy contributions context sensitive and not context free, 
and we would lose one entire level of our complexity hierarchy (\cref{fig:mainresult}).  
The increase in complexity caused by a binary encoding has also been observed from a different angle in Ref.\ \cite{St21}. 

\subsection{The complexity measure provided by the Chomsky hierarchy}
\label{ssec:Chomsky}

We now classify ${\sf L}_{\mc{M}}$ in the Chomsky hierarchy. 
To that end, we define two additional levels of the Chomsky hierarchy which are obtained by posing restrictions on the automata that accept context free and context sensitive languages, namely constructive context free and constructive context sensitive.
We will conjecture that these two new levels are identical with context free and context sensitive, respectively.

\begin{definition}[Constructive automaton]\label{def:constructive}
Let $\mc{M}$ be an Ising model and ${\sf L}_{\mc{M}}$ be its language. 

\begin{enumerate}[label=(\roman*), ref=(\roman*)]
\item 
A PDA $P$ that decides ${\sf L}_{\mc{M}}$ is called \emph{constructive} if 
for any $n \in N_{\mc{M}}$,  there exists a unique partition of edges
$(E_{\mc{M}})_n = \bigcup_{m=1}^{r_n}I_m $,
such that on well-formed inputs $s_1 \ldots s_n \bullet \sigma^k$, $P$ operates as follows: 
\begin{enumerate}[label=(\alph*), ref=(\alph*)]
\item \label{def:constructive:i}
First, $P$ accumulates $u (H_{\mc{M}}(s_1\ldots s_n))$ on its stack. 
$P$ iterates over $m=1, \ldots, r_n$, and in each step of the iteration
it stores the states of the spins that are contained in edges from $I_m$, i.e.\ $V_n\vert_{I_m}$
in its internal states. 
$P$ then adds the unary encoding of the energy contribution of these spins 
\begin{align}
   H_{\mc{M}}\vert_{I_m} \coloneqq  \sum _{(i,j) \in I_m}h(s_i,s_j)
\end{align} 
to its stack.
\item \label{def:constructive:ii}
Second, $P$ compares its stack content, $u (H_{\mc{M}}(s_1\ldots s_n))$ to the input energy $\sigma^k$ and accepts if and only if they are equal.
\end{enumerate}

\item 
A LBA  $M$ that decides ${\sf L}_{\mc{M}}$ is called \emph{constructive} if it uses a designated energy tape $T_e$ to accumulate the energy $H_{\mc{M}}(s_1\ldots s_n)$ 
in binary, 
compares the content of $T_e$ to the input energy and accepts if and only if the two values coincide.
\end{enumerate}
\end{definition}

For the constructive PDA, note that there exists an upper bound for the size of $I_m$, 
as by definition a constructive PDA must be able to store all spin states contained in $V_n\vert_{I_m}$ in its states. 

Further, note that the definition of constructive LBA uses the fact that multi-tape LBAs are equivalent to single-tape LBAs in terms of the set of languages they accept. 

In words, a constructive automaton works the way one naively expects: 
it adds up local energy contributions in a pre-determined way, 
and then compares the result to the input energy. 
In particular, constructive automata compute $H_{\mc{M}}$ as a function.
Working with constructive automaton thus ensures that we consider the function problem of computing $H_{\mc{M}}$, although for technical reasons, we formulate it as a decision problem (deciding the function graph of $H_{\mc{M}}$) so that we can work with the Chomsky hierarchy. 
We  conjecture that the constructive condition on the automata is not necessary: 

\begin{conjecture}\label{conj}
For every context free ${\sf L}_{\mc{M}}$  there exists a constructive PDA that decides it. 
For every context sensitive ${\sf L}_{\mc{M}}$  there exists a constructive LBA that decides it. 
\end{conjecture}

In this work, we do not assume that this conjecture is true, i.e.\ we explicitly state whenever we require that an automaton is constructive. 

Constructive PDA and constructive LBA define two new complexity levels for ${\sf L}_{\mc{M}}$. 
If ${\sf L}_{\mc{M}}$ is accepted by a constructive PDA, 
we say that ${\sf L}_{\mc{M}}$ is constructive context free; 
if  ${\sf L}_{\mc{M}}$ is accepted by a constructive LBA, 
we say that ${\sf L}_{\mc{M}}$ is constructive context sensitive (cf.\ \cref{fig:mainresult}).
Considering only languages of the type ${\sf L}_{\mc{M}}$, 
constructive context free is a subset of context free, 
and constructive context sensitive is a subset of context sensitive. 
We now show that supplementing the Chomsky hierarchy with these two complexity levels still forms a hierarchy, i.e.\ that regular is a subset of constructive context free and context free is a subset of constructive context sensitive.   

\begin{proposition}[Refined Chomsky hierarchy]
Let $\mc{M}$ be an Ising model and ${\sf L}_{\mc{M}}$ be its language. 
\begin{enumerate}[label=(\roman*)]
 \item\label{it:FineGrain1} If ${\sf L}_{\mc{M}}$ is regular then it is constructive context free.
 \item\label{it:FineGrain2} If ${\sf L}_{\mc{M}}$ is context free then it is constructive context sensitive.
\end{enumerate}
\end{proposition}

\begin{proof}
Starting with \ref{it:FineGrain1}, if ${\sf L}_{\mc{M}}$ is regular then by \cref{thm:main} \ref{thm:main:regular}, $\mc{M}$ is finite, i.e.\ there exists a maximum system size. Hence, we can build a PDA $P$ that on input $s_1 \ldots s_n \bullet \sigma^k$
reads and stores all spin symbols. As there are   finitely many, this can be done with a finite number of states. Next $P$ adds the entire energy $H_{\mc{M}}(s_1 \ldots s_n)$ to its stack. This can be hardwired in the transition rules. Finally, $P$ compares its stack to 
the input energy $\sigma^k$ and accepts if and only if they are equal. Note that $P$ is trivially constructive. 
For all system sizes, the partition of edges consists of one element only, namely the entire edge set.  

Next we prove \ref{it:FineGrain2}. 
As ${\sf L}_{\mc{M}}$ is context free, it has a context free grammar in Greibach normal form \cite[Lecture 21]{Ko97}. 
From this grammar, 
one can build a PDA  $P$ without $\epsilon$-transitions that decides this language \cite[Lecture 24]{Ko97}. 
Each transition rule of $P$ pushes at most $k$ symbols to the stack. 
As there are no $\epsilon$-transitions, $P$ uses an amount of stack memory that is linear in the length of the input.
$P$ can hence be simulated by a LBA $M$ simply by using an additional tape to simulate the linear bounded stack of $P$. 
We now show that this LBA can be assumed to be constructive.  
First note that when processing any well-formed input $s_1 \ldots s_n \bullet \sigma^k$, 
once the head of $P$ reaches $\bullet$,
$P$ has stored the energy $H_{\mc{M}}(s_1 \ldots s_n)$ either on its stack, in its states or by using a combination of both, as otherwise $P$ could not decide if $u(H_{\mc{M}}(s_1 \ldots s_n))=\sigma^k$. 
It follows that, when $M$ simulates $P$,  $M$ stores $H_{\mc{M}}(s_1 \ldots s_n)$ at some point during the computation, and we can w.l.o.g.\ assume that $M$ stores this energy in binary, which makes it constructive.  
\end{proof}

Finally, we can define the complexity measure for Ising models: 
The complexity of an Ising model is obtained by classifying its language in the (refined) Chomsky hierarchy. 
The latter induces a complexity hierarchy of Ising models themselves.

\begin{definition}[Complexity measure]
    \label{def:Complexity}
Let $\mc{M}$ be an Ising model. We say that $\mc{M}$ has complexity $X$ if its language ${\sf L}_{\mc{M}}$ belongs to level $X$ of the (refined) Chomsky hierarchy.
Here $X$ can be any of the following alternatives: 
regular, 
constructive context free, 
context free, 
constructive context sensitive, 
context sensitive, 
decidable. 
\end{definition}

\section{Full Classification of Ising models}
\label{sec:mainresult}

In this section we state and discuss our main result: a full classification of the complexity of   Ising models based on properties of their interaction graphs (\cref{thm:main}).
First we define the relevant properties of interaction graphs, more precisely of families of interaction graphs (\cref{ssec:PropIsing}).  
Then we state \cref{thm:main} (\cref{ssec:main}), and provide a proof in \cref{sec:proofmain}.

\subsection{Properties of Ising models}
\label{ssec:PropIsing}

Let us now introduce several properties of families of interaction graphs. 
These properties can be divided into two classes. 
The first class captures the complexity of the family of interaction graphs.  
The second class captures how certain properties of the individual graphs contained therein scale with the system size.

In order to quantify the complexity of the family of interaction graphs of an Ising model $\mc{M}$ 
we define the edge language of $\mc{M}$, ${\sf E}_{\mc{M}}$, as a formal language that encodes the entire family of interaction graphs of $\mc{M}$. 
Consequently, classifying ${\sf E}_{\mc{M}}$ in the Chomsky hierarchy measures the complexity of the family of interaction graphs of $\mc{M}$.

\begin{definition}[Edge language]
\label{def:EdgeLanguage}
Let $\mc{M}$ be an Ising model. The \emph{edge language} of $\mc{M}$, ${\sf E}_{\mc{M}}$,  is defined as 
\begin{align}\label{eqn:EdgeLanguage}
    {\sf E}_{\mc{M}} \coloneqq \{ 0^{i-1}10^{j-i-1}10^{n-j} 
    \mid  n \in N_{\mc{M}}, \: (i,j) \in (E_{\mc{M}})_n \} 
\end{align}
\end{definition}

In words, 
${\sf E}_{\mc{M}}$ directly encodes the interaction graphs of $\mc{M}$, as $w=w_1\ldots w_n \in {\sf E}_{\mc{M}}$ if and only if $n \in N_{\mc{M}}$ and $w$ is a row of the incidence matrix of the graph defined by $(E_{\mc{M}})_n$.
So each string $w$ in the edge language specifies  one edge of one interaction graph of $\mc{M}$, 
as well as the system size of that interaction graph (encoded in the length of $w$).

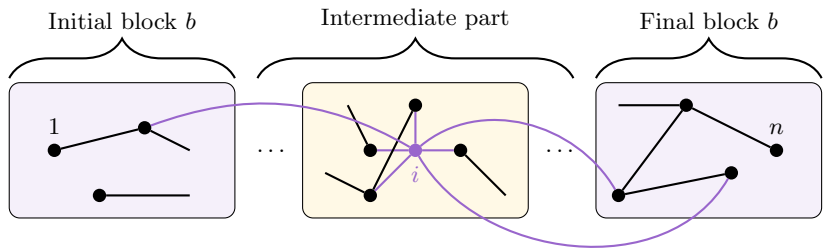
\begin{figure*}[t]\centering
    \begin{center}
        \begin{tikzpicture}[->,>=stealth',auto, scale=0.6,
            thick, main node/.style={circle,draw, fill = black, inner sep = 1.5}]{
        \node[main node] (1) at (0,0) {};
        \node[main node] (2) at (0,1) {};
        \node[main node, amethyst, label=below:{\color{amethyst} $i$}] (3) at (1,1) {};
        \node[main node] (4) at (2,1) {};
        \node[main node] (5) at (1,2) {};
        \coordinate (6) at (3,0);
        \coordinate (7) at (-1,0.5);
        \coordinate (8) at (-0.5,2);

        \draw[-, amethyst] (3) -- (2) ; 
        \draw[-, amethyst] (3) -- (1) ; 
        \draw[-] (6) -- (4); 
        \draw[-, amethyst] (3) -- (4); 
        \draw[-] (1) -- (5);
        \draw[-, amethyst] (3) -- (5);
        \draw[-] (1) -- (7);
        \draw[-] (2) -- (8);

	\begin{scope}[on background layer]
        \draw[rounded corners, fill=amber!10] (-1.5,-0.5) rectangle ++(5,3);

        \draw[rounded corners, fill=amethyst!10] (-8,-0.5) rectangle ++(5,3);
        \draw[rounded corners, fill=amethyst!10] (5,-0.5) rectangle ++(5,3);       
        \end{scope}

	\node (ldots1) at (-2.2,1) {$\ldots$};
	\node (ldots2) at (4.2,1) {$\ldots$};

        \node[main node, label=above:{$1$}] (9) at (-7,1) {};
        \node[main node] (10) at (-5,1.5) {};
        \node[main node] (11) at (-6,0) {};
        \coordinate (12) at (-4,0);
        \coordinate (13) at (-4,1);    

        \draw[-] (11) -- (12);
        \draw[-] (9) -- (10);
        \draw[-] (10) -- (13);

        \draw[-, amethyst] (3) to [out=150, in=20] (10);

        \draw[-, decorate,decoration={brace,amplitude=15pt},xshift=0pt,yshift=0pt]
              (-2.5,2.6) -- (4.5,2.6) node [black,midway, yshift=15pt] {Intermediate part};

        \draw[-, decorate,decoration={brace,amplitude=15pt},xshift=0pt,yshift=0pt]
              (-8,2.6) -- (-3,2.6) node [black,midway, yshift=15pt] {Initial block $b$};

        \draw[-, decorate,decoration={brace,amplitude=15pt},xshift=0pt,yshift=0pt]
              (5,2.6) -- (10,2.6) node [black,midway, yshift=15pt] {Final block $b$};

        \node[main node] (14) at (5.5,0) {};
        \node[main node] (15) at (7,2) {};
        \node[main node] (16) at (8,0.5) {};
        \node[main node, label=above:{$n$}] (17) at (9,1) {};
        \coordinate (18) at (5.5,2) {};

        \draw[-] (18) -- (15); 
        \draw[-] (14) -- (16);
        \draw[-] (14) -- (15);
        \draw[-] (15) -- (17);
        \draw[-, amethyst] (3) to [out=40, in=120] (14);
        \draw[-, amethyst] (3) to [out=300, in=240] (16);

        }
        \end{tikzpicture}
        \end{center}
\caption{Interaction graph of a limited Ising model (\cref{def:PropIsing} \ref{def:limited}). 
Given any vertex $i$ not within the first or last $b$ vertices, 
all incident edges are either connected to a vertex $j$ with $\vert j-i\vert \leq b$, i.e.\ to a vertex within the same block  (pale yellow block), or to one of the first or last $b$ vertices (pale purple blocks).}
\label{fig:limited}
\end{figure*}

The following proposition justifies the classification of the complexity of Ising models based on their edge language, as it states that not only 
${\sf L}_{\mc{M}}$ but also ${\sf E}_{\mc{M}}$ characterise $\mc{M}$ uniquely.

\begin{proposition}[Uniqueness of the edge language]
\label{thm:UniqueLanguage}
Let $\mc{M}$ and $\mc{M}^{\prime}$ be two Ising models. 
The following are equivalent: 
\begin{enumerate}[label=(\roman*)]
    \item\label{en:MEq} $\mc{M} = \mc{M}^{\prime}$
    \item\label{en:LEq} ${\sf L}_{\mc{M}} = {\sf L}_{\mc{M}^{\prime}}$
    \item\label{en:EEq} ${\sf E}_{\mc{M}} = {\sf E}_{\mc{M}^{\prime}}$
\end{enumerate}
\end{proposition}

\begin{proof}
The two implications \ref{en:MEq} $\implies$ \ref{en:LEq} and \ref{en:EEq} $\implies$ \ref{en:MEq} are obvious. 
We thus need to show \ref{en:LEq} $\implies$ \ref{en:EEq}. 
If ${\sf L}_{\mc{M}}= {\sf L}_{\mc{M}^{\prime}}$ then clearly $N_{\mc{M}}= N_{\mc{M}^{\prime}}$. Now take any $n \in N_{\mc{M}}$. 
To show that  $(E_{\mc{M}})_n = (E_{\mc{M}^{\prime}})_n$ consider the  function
\begin{align}\label{eqn:corr}
    \begin{aligned}
    C_{\mc{M}}(n,i,j) \coloneqq -\frac{1}{4} \bigl[&H_{\mc{M}}(0^n) 
  + H_{\mc{M}}(0^{i-1}10^{j-i-1}10^{n-j})\\
    & - H_{\mc{M}}(0^{i-1}10^{n-i}) 
     - H_{\mc{M}}(0^{j-1}10^{n-j})\bigr ]
    \end{aligned}
\end{align}
Using \cref{eqn:IsingHamiltonian} one readily concludes that 
\begin{align}
    C_{\mc{M}}(n,i,j) = \begin{cases}
        1 & \text{if} \ (i,j) \in (E_{\mc{M}})_n \\
        0 & \text{else}
    \end{cases}
\end{align}
Thus 
\begin{equation}
(E_{\mc{M}})_n = \{ (i,j) \mid C_{\mc{M}}(n,i,j) = 1 \}
\end{equation}
But as ${\sf L}_{\mc{M}}= {\sf L}_{\mc{M}^{\prime}}$ is equivalent to $H_{\mc{M}}= H_{\mc{M}^{\prime}}$, it is also the case that $C_{\mc{M}}(n,i,j) = C_{\mc{M}^{\prime}}(n,i,j)$,  and hence it follows that $(E_{\mc{M}})_n = (E_{\mc{M}^{\prime}})_n$.
\end{proof}

Next we consider how certain properties of the interaction graphs of an Ising model scale, i.e.\ how they change when changing the system size.
\begin{definition}[Finite, limited and linear Ising model]
    \label{def:PropIsing}
    We call an Ising model $\mc{M}$
    \begin{enumerate}[label=(\roman*)]
        \item \label{def:finite}
        \emph{finite} if $N_{\mc{M}}$ is finite.
        \item \label{def:limited}
        \emph{limited} if there exists a natural number $b$ such that for all $n \in N_{\mc{M}}$ and any $(i,j) \in (E_{\mc{M}})_n$, 
        if $j-i > b$ then either $i \leq b$ or $n-j \leq b$. 
        \item \label{def:linear}
        \emph{linear} if there exists a natural number $k$ such that for all $n\in N_{\mc{M}}$, it is the case that $\vert (E_{\mc{M}})_n \vert \leq k   n$.
    \end{enumerate}
\end{definition}

For every Ising model the number of edges scales at most quadratically with the system size (because the complete graph has $\binom{n}{2}$ edges). 
The properties \ref{def:finite} finite and \ref{def:linear} linear fine-grain the scaling of the number of edges---in a truncated scaling (finite), and a linear scaling. 

Property \ref{def:limited} (limited) 
captures the scaling of the maximal interaction range. 
Intuitively, an Ising model is limited if there is an upper bound on its interaction range, 
i.e.\ a $b\in \mathbb{N}$ such that there is no edge $(i,j)$ where $i$ and $j$ are separated by at least $b$ other vertices (i.e.\ $\vert j-i\vert >b$). 
With an exception:  if  $i$ is within the first $b$ vertices or $j$ is within the last $b$ vertices, 
then either of them can have long-range edges, that is, they can be linked to other vertices which are further away than $b$   (see \cref{fig:limited}). 

Finally, the properties finite, limited, linear form a hierarchy. 
Clearly, every finite Ising model is limited.
Proving that every limited Ising model is linear follows from a simple counting argument: 
If $\mc{M}$ is limited, considering $n>2b$ vertices, 
the first and last $b$ vertices are included in at most  $2b  n$ edges, 
and the remaining vertices are included in at most  $(n-2b)  2b$ additional edges. 
In total, we have 
\begin{equation}
\vert (E_{\mc{M}})_n \vert \leq 4 b    n - 4 b^2 
\end{equation}
As per definition  $b$ is independent of the system size, this
shows that $\mc{M}$ is linear.

Overall we have two hierarchies that capture properties of the family of interaction graphs of an Ising model (\cref{fig:mainresult2}). 
The first hierarchy classifies Ising models based on the complexity of their family of interaction graphs; 
specifically, it encodes the interaction graphs as a language ${\sf E}_{\mc{M}}$ and classifies it in the Chomsky hierarchy (purple shapes of \cref{fig:mainresult2}). 
The second hierarchy classifies Ising models based on the scaling of the number of edges, 
as well as the scaling of the maximal interaction range with the system size (red shapes of \cref{fig:mainresult2}). 

\begin{figure*}[t]\centering
    \includegraphics[width=1.2\columnwidth]{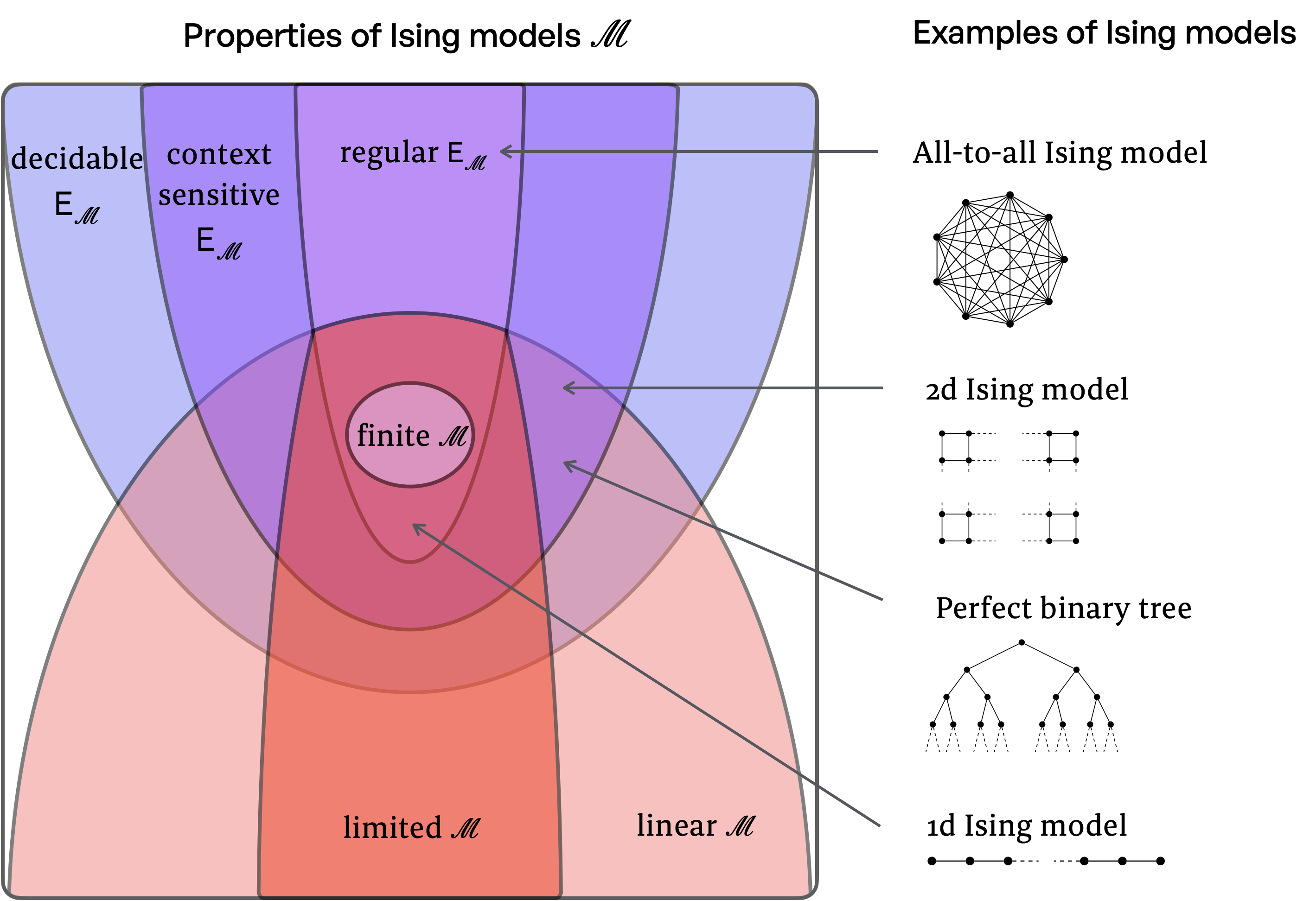}
    \caption{Set diagram of the properties of Ising models. 
    The purple shapes correspond to properties that capture the complexity of the family of interaction graphs, phrased in terms of the complexity of ${\sf E}_{\mc{M}}$ in the Chomsky hierarchy.
    The red shapes correspond to properties that capture the scaling of the connectivity and the number of long-range interactions with the system size. 
    The right-hand side shows where several examples are located in this diagram. 
    Note that both the 2d Ising model and the all-to-all Ising model are constructive context sensitive; yet, 
    the all-to-all model has regular ${\sf E}_{\mc{M}}$ and is not linear (and thus fails to be constructive context free), whereas the 2d model has context sensitive ${\sf E}_{\mc{M}}$ and is linear. 
    }
    \label{fig:mainresult2}
\end{figure*}

\subsection{Main result}
\label{ssec:main}

We are now ready to state the full classification of the complexity of Ising models based on the properties introduced above.

\begin{theorem}[Main result] \label{thm:main}  
Let $\mc{M}$ be an Ising model. 
\begin{enumerate}[label=(\roman*)]
\item 
\label{thm:main:regular}
${\sf L}_{\mc{M}}$ is regular if and only if $\mc{M}$ is finite; 

\item 
\label{thm:main:cf1}
${\sf L}_{\mc{M}}$ is  constructive context free if and only if $\mc{M}$ is linear and  ${\sf E}_{\mc{M}}$ is regular; 

\item 
\label{thm:main:ccs}
${\sf L}_{\mc{M}}$ is constructive context sensitive if and only if  ${\sf E}_{\mc{M}}$ is context sensitive; 

\item 
\label{thm:main:undec}
${\sf L}_{\mc{M}}$ is decidable if and only if ${\sf E}_{\mc{M}}$ is decidable.
\end{enumerate}
\end{theorem}

While this theorem is proven in \cref{sec:proofmain}, the statements can be intuitively understood as follows. 

\paragraph*{\ref{thm:main:regular}.} 
A finite Ising model only contains a finite number of system sizes, thus both ${\sf L}_{\mc{M}}$ and ${\sf E}_{\mc{M}}$ are finite languages, which are trivially regular.  
Conversely, by the pumping lemma for regular languages, it follows that infinite Ising models cannot have a regular language.

\paragraph*{\ref{thm:main:cf1}.}
The essence of the argument is the following.  
If $\mc{M}$ is linear and ${\sf E}_{\mc{M}}$ is regular, a constructive PDA for ${\sf L}_{\mc{M}}$ can be built from a FSA for ${\sf E}_{\mc{M}}$.
Deciding ${\sf L}_\mc{M}$ amounts to deciding ${\sf E}_\mc{M}$ and adding the individual energy contributions.
The PDA decides ${\sf E}_\mc{M}$ by running the FSA in its states. As $\mc{M}$ is linear, also adding the individual energy contributions is possible for a PDA.
Note that  if $\mc{M}$ is not finite, adding the individual energy contributions requires at least a PDA, i.e.\ cannot be achieved by a FSA. 
Consequently, deciding the edges can at most require a FSA, since PDAs can use FSAs but not PDAs as subroutines. 
Conversely, if $\mc{M}$ were not linear, then already adding the individual energy contributions would require a LBA (independently of the complexity of ${\sf E}_\mc{M}$), as  it is the case for the all-to-all Ising model (\cref{fig:mainresult2} and \cref{ssec:all2all}). 
Regularity of ${\sf E}_{\mc{M}}$ can be proven by using the constructive PDA for ${\sf L}_{\mc{M}}$  to compute \cref{eqn:corr}. So this PDA can be modified to decide ${\sf E}_{\mc{M}}$. 
From the fact that the PDA is constructive it then follows that the modified PDA only ever uses a single cell of stack memory, so it effectively is a FSA.

\paragraph*{\ref{thm:main:ccs}.}
In contrast to PDAs, LBAs can use LBAs as subroutines. 
This is the key difference between \ref{thm:main:cf1} and \ref{thm:main:ccs}. 
It follows that there is no separation in the complexity of ${\sf E}_\mc{M}$ and ${\sf L}_\mc{M}$, 
i.e.\ in contrast to \ref{thm:main:cf1}, they can be of the same complexity. 
The constructive LBA that decides ${\sf L}_{\mc{M}}$ can be built by using a LBA for deciding ${\sf E}_{\mc{M}}$, to select the individual edges that contribute to the energy, 
and another LBA to sum up these individual energy contributions. 
Conversely, if ${\sf L}_{\mc{M}}$ is constructive context sensitive then a LBA for ${\sf E}_{\mc{M}}$ can be built by modifying the LBA for ${\sf L}_{\mc{M}}$ so that it computes \cref{eqn:corr}.

\paragraph*{\ref{thm:main:undec}.}
Turing machines can also use Turing machines as subroutines. A Turing machine that decides ${\sf L}_{\mc{M}}$ can be built in the same way as the LBA in the previous case, and also the converse direction of \ref{thm:main:undec} works the same way as that of \ref{thm:main:ccs}.

Let us highlight a corollary of the proof of \cref{thm:main}, which will prove useful in the examples of \cref{sec:examples}: 
\begin{corollary} \label{cor:limited} 
If $\mc{M}$ is not limited then ${\sf L}_{\mc{M}}$ is not constructive context free. 
\end{corollary}

\begin{proof}
From \cref{thm:main} \ref{thm:main:cf1} we know that if ${\sf L}_{\mc{M}}$ is constructive context free then $\mc{M}$ is linear and ${\sf E}_{\mc{M}}$ is regular. 
From \cref{ssec:proofmain2} \ref{thm:main:cf1:iii}, it follows that $\mc{M}$ is limited. The corollary states the contrapositive. 
\end{proof}

\section{Examples: The complexity of Ising models}
\label{sec:examples}

We now consider various Ising models and compute their complexity by applying \cref{thm:main}. 
Specifically, we consider the $1$d Ising model (\cref{ssec:1d}) 
the $1$d Ising model with periodic boundary conditions (\cref{ssec:1dperiodic}), 
the Ising model on ladder graphs (\cref{ssec:ladder}),
the Ising model  on layerwise complete graphs (\cref{ssec:network}),
the $2$d Ising model (\cref{ssec:2D}),
the all-to-all Ising model  (\cref{ssec:all2all}),
and the Ising model on perfect binary trees (\cref{ssec:tree}). 
The results are summarised in \cref{tab:examples}. 

\begin{table*}[th]\centering
    \begin{tabular}{l|l|l|l|l|l}
    Ising model $\mc{M}$ & ${\sf L}_{\mc{M}}$ & ${\sf E}_{\mc{M}}$ & Finite  & Limited  & Linear  \\
    \hline
      $\mc{M}_{{\mr{1d}}}$ & Constructive context free & Regular & No & Yes & Yes \\
    $\mc{M}_{{\mr{circ}}}$ &  Constructive context free & Regular & No & Yes & Yes \\
     $\mc{M}_{{\mr{ladder}}}$ &  Constructive context free & Regular & No & Yes & Yes \\
     $\mc{M}_{{\mr{layer}}}$  &  Constructive context free & Regular & No & Yes & Yes \\
     $\mc{M}_{{\mr{2d}}}$ & Constructive context sensitive & Context sensitive & No & No & Yes \\
      $\mc{M}_{{\mr{all}}}$ &  Constructive context sensitive & Regular & No & No & No \\
     $\mc{M}_{{\mr{tree}}}$ &  Constructive context sensitive &  Context sensitive & No & No & Yes \\
    \end{tabular}
    \caption{The Ising models (first column) considered in \cref{sec:examples}, their complexity (second column) and 
    the properties that determine their complexity (remaining columns). 
    }
    \label{tab:examples}
\end{table*}

\subsection{1d Ising model}
\label{ssec:1d}

The most straightforward example of a constructive context free Ising model uses $1$-dimensional chains as interaction graphs (\cref{fig:1d}).
We denote this model as $\mc{M}_{\mr{1d}}\coloneqq (N_{\mr{1d}}, E_{\mr{1d}})$, defined by
\begin{align}
    \begin{aligned}
        N_{\mr{1d}}&\coloneqq \{n \in \mathbb{N} \mid n \geq 2 \} \\
        (E_{\mr{1d}})_n&\coloneqq \bigl\{ (i,i+1) \mid 1 \leq i \leq n-1\bigr\}
    \end{aligned}
\end{align}
We now use \cref{thm:main} \ref{thm:main:cf1} to prove that the language of $\mc{M}_{\mr{1d}}$ is constructive context free. To that end, let us show that it is linear, and its edge language is regular. 
Linearity is immediate, as for any $n \in N_{\mr{1d}}$, $\vert (E_{\mr{1d}})_n \vert = n-1$.
Also regularity of ${\sf E}_{\mr{1d}}$ can be concluded straightforwardly, as ${\sf E}_{\mr{1d}} = 0^*110^*$  (where $0^*$ denotes the concatenation of any finite number of 0s, including the empty one). 
Since $\mc{M}_{\mr{1d}}$ is clearly not finite, 
by \cref{thm:main} \ref{thm:main:regular} ${\sf L}_{\mr{1d}}$ is not regular.

\subsection{1d Ising model with periodic boundary conditions}
\label{ssec:1dperiodic}
A second Ising model with constructive context free language is obtained by taking circles  as interaction graphs (\cref{fig:circle}).
We denote this model as $\mc{M}_{\mr{circ}}\coloneqq (N_{\mr{circ}},E_{\mr{circ}})$, where
\begin{align}
    \begin{aligned}
        N_{\mr{circ}}&\coloneqq \{n \in \mathbb{N} \mid n \geq 3 \}\\
        (E_{\mr{circ}})_n &\coloneqq (E_{\mr{1d}})_n \cup \{(1,n)\} 
    \end{aligned}
\end{align}
Again, as $\vert (E_{\mr{circ}})_n\vert = n$, ${\mc M}_{\mr{circ}}$ clearly is linear. Besides,  
\begin{equation}
{\sf E}_{\mr{circ}}= {\sf E}_{\mr{1d}} \cup 100^*1 
\end{equation} 
is a union of regular languages, so it is regular. So according to \cref{thm:main} \ref{thm:main:cf1} ${\sf L}_{\mr{circ}}$ is constructive context free. 
As $\mc{M}_{\mr{circ}}$  is not finite, according to \cref{thm:main} \ref{thm:main:regular} ${\sf L}_{\mr{circ}}$ is not regular.

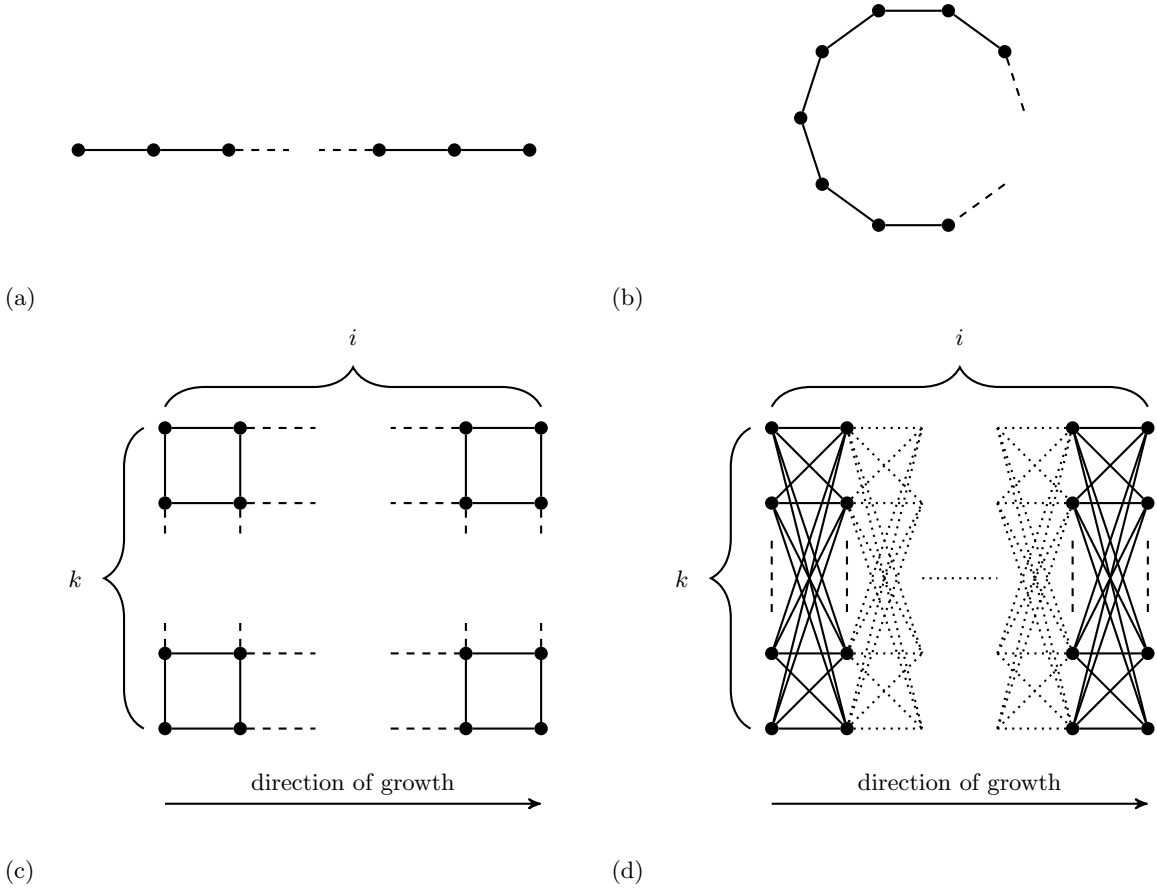
\begin{figure*}[t]
    \centering
    \begin{subfigure}[b]{.45\textwidth}
        \centering
            \vspace{1cm}
            \begin{tikzpicture}[->,>=stealth',auto,
              thick,main node/.style={circle,draw,fill=black, inner sep=1.5pt}]{
            \node[main node,label=above:{}] (1) at (-1,0) {};
            \node[main node,label=above:{}] (2) at (0,0) {};
            \node[main node,label=above:{}] (3) at (1,0) {};
            \node[main node,label=above:{}] (6) at (3,0) {};
            \node[main node,label=above:{}] (7) at (4,0) {};
            \node[main node,label=above:{}] (8) at (5,0) {};
            
            \coordinate (9) at (1.8,0);
            \coordinate (10) at (2.2,0);
            
            \draw [-] (1) -- (2);
            \draw [-] (2) -- (3);
            \draw [-] (6) -- (7);
            \draw [-] (7) -- (8);
            
            \draw [dashed, -] (3) -- (9); 
            \draw [dashed, -] (6) -- (10);
            }
            \end{tikzpicture}
            \vspace{1.5cm}
             \caption{}
            \label{fig:1d}
    \end{subfigure}%
    \begin{subfigure}[b]{.45\textwidth}
        \centering
            \begin{tikzpicture}[->,>=stealth',auto,
              thick,main node/.style={circle,draw,fill=black,inner sep=1.5pt}]{
            
            \foreach \a in {1,2,...,8}{
            \node[main node] (\a) at (\a*360/10: 1.5cm) {};
            }
            \coordinate (9) at (9*360/10: 1.5cm);
            \coordinate (10) at (360: 1.5cm);
            
            \draw [-] (1) -- (2);
            \draw [-] (3) -- (2);
            \draw [-] (3) -- (4);
            \draw [-] (4) -- (5);
            \draw [-] (5) -- (6);
            \draw [-] (7) -- (6);
            \draw [-] (7) -- (8);
            \draw [-, dashed] (1) -- (10);
            \draw [-, dashed] (9) -- (8);
            }
            \end{tikzpicture}
            \vspace{0.5cm}
             \caption{}
            \label{fig:circle}
    \end{subfigure}
    \begin{subfigure}[t]{.45\textwidth}
        \centering
            \begin{tikzpicture}[->,>=stealth',auto,
                thick,main node/.style={circle,draw,fill=black,inner sep=1.5pt}]{
              
              \node[main node] (1) at (0,0) {};
              \node[main node] (2) at (1,0) {};
              \node[main node] (3) at (0,-1) {};
              \node[main node] (4) at (1,-1) {};
              \coordinate (5) at (2,0);
              \coordinate (6) at (2,-1);
              
              \draw[-] (1) -- (2);
              \draw[-] (1) -- (3);
              \draw[-] (2) -- (4);
              \draw[-] (3) -- (4);
              
              \draw[-, dashed] (2) -- (5);
              \draw[-, dashed] (4) -- (6);
              
              \node[main node] (7) at (4,0) {};
              \node[main node] (8) at (5,0) {};
              \node[main node] (9) at (4,-1) {};
              \node[main node] (10) at (5,-1) {};
              \coordinate (11) at (3,0);
              \coordinate (12) at (3,-1);
              
              \draw[-] (7) -- (8);
              \draw[-] (7) -- (9);
              \draw[-] (8) -- (10);
              \draw[-] (9) -- (10);
              
              \draw[-, dashed] (7) -- (11);
              \draw[-, dashed] (9) -- (12);

              \coordinate (13) at (0,-1.5);
              \coordinate (14) at (1,-1.5);
              \coordinate (15) at (4,-1.5);
              \coordinate (16) at (5,-1.5);
              
              \draw[-, dashed] (3) -- (13);
              \draw[-, dashed] (4) -- (14);
              \draw[-, dashed] (9) -- (15);
              \draw[-, dashed] (10) -- (16);
                    
              \node[main node] (17) at (0,-3) {};
              \node[main node] (18) at (1,-3) {};
              \node[main node] (19) at (0,-4) {};
              \node[main node] (20) at (1,-4) {};
                        
              \node[main node] (21) at (4,-3) {};
              \node[main node] (22) at (5,-3) {};
              \node[main node] (23) at (4,-4) {};
              \node[main node] (24) at (5,-4) {};
              
              \draw[-] (17) -- (18);
              \draw[-] (17) -- (19);
              \draw[-] (18) -- (20);
              \draw[-] (19) -- (20);
              
              \draw[-] (21) -- (22);
              \draw[-] (21) -- (23);
              \draw[-] (22) -- (24);
              \draw[-] (23) -- (24);
              
              \coordinate (25) at (0,-2.5);
              \coordinate (26) at (1,-2.5);
              \coordinate (27) at (4,-2.5);
              \coordinate (28) at (5,-2.5);
              
              \draw[-, dashed] (17) -- (25);
              \draw[-, dashed] (18) -- (26);
              \draw[-, dashed] (21) -- (27);
              \draw[-, dashed] (22) -- (28);
              
              \coordinate (29) at (2,-3);
              \coordinate (30) at (3,-3);
              \coordinate (31) at (2,-4);
              \coordinate (32) at (3,-4);
              
              \draw[-, dashed] (18) -- (29);
              \draw[-, dashed] (30) -- (21);
              \draw[-, dashed] (20) -- (31);
              \draw[-, dashed] (32) -- (23);
              
              \draw[-, decorate,decoration={brace,amplitude=15pt},xshift=-8pt,yshift=0pt]
              (0,-4) -- (0,0) node [black,midway, xshift=-20pt] {$k$};

              \draw[-, decorate,decoration={brace,amplitude=15pt},xshift=0pt,yshift=8pt]
              (0,0) -- (5,0) node [black,midway, yshift=20pt] {$i$};
              
              \draw[->] (0,-5) -- (5,-5) node [midway] {direction of growth};  
              }
              \end{tikzpicture}
              \vspace{0.5cm}
            \caption{}\label{fig:ladder}
    \end{subfigure}%
    \begin{subfigure}[t]{.45\textwidth}
        \centering
            \begin{tikzpicture}[->,>=stealth',auto,
                thick,main node/.style={circle,draw,fill=black,inner sep=1.5pt}]{
              
              \node[main node] (1) at (0,0) {};
              \node[main node] (2) at (0,-1) {};
              \node[main node] (3) at (0,-3) {};
              \node[main node] (4) at (0,-4) {};
              \draw[-, dashed] (0,-1.5) -- (0,-2.5);
              
              \node[main node] (5) at (1,0) {};
              \node[main node] (6) at (1,-1) {};
              \node[main node] (7) at (1,-3) {};
              \node[main node] (8) at (1,-4) {};
              \draw[-, dashed] (1,-1.5) -- (1,-2.5);
     
              \foreach \x in {1,...,4}
                  \foreach \y in {5,...,8}  
                    \draw[-] (\x)--(\y);
                    
              \node[main node] (9) at (4,0) {};
              \node[main node] (10) at (4,-1) {};
              \node[main node] (11) at (4,-3) {};
              \node[main node] (12) at (4,-4) {};
              \draw[-, dashed] (4,-1.5) -- (4,-2.5);
              
              \node[main node] (13) at (5,0) {};
              \node[main node] (14) at (5,-1) {};
              \node[main node] (15) at (5,-3) {};
              \node[main node] (16) at (5,-4) {};
              \draw[-, dashed] (5,-1.5) -- (5,-2.5);
              
              \foreach \x in {9,...,12}
                  \foreach \y in {13,...,16}  
                    \draw[-] (\x)--(\y);
                    
              \coordinate (17) at (2,0);
              \coordinate (18) at (2,-1);
              \coordinate (19) at (2,-3);
              \coordinate (20) at (2,-4);
              
              \foreach \x in {5,...,8}
                  \foreach \y in {17,...,20}  
                    \draw[-, dotted] (\x)--(\y);
                    
              \coordinate (21) at (3,0);
              \coordinate (22) at (3,-1);
              \coordinate (23) at (3,-3);
              \coordinate (24) at (3,-4);
              
              \foreach \x in {21,...,24}
                  \foreach \y in {9,...,12}  
                    \draw[-, dotted] (\x)--(\y);
                    
              \draw[-, dotted] (2,-2) -- (3,-2);

              \draw[-, decorate,decoration={brace,amplitude=15pt},xshift=-8pt,yshift=0pt]
              (0,-4) -- (0,0) node [black,midway, xshift=-20pt] {$k$};

              \draw[-, decorate,decoration={brace,amplitude=15pt},xshift=0pt,yshift=8pt]
              (0,0) -- (5,0) node [black,midway, yshift=20pt] {$i$};
              
              \draw[->] (0,-5) -- (5,-5) node [midway] {direction of growth};
              }
            \end{tikzpicture}
            \vspace{0.5cm}
            \caption{}\label{fig:network}
    \end{subfigure}
    \caption{Interaction graphs of $\mc{M}_{\mr{1d}}$ (\ref{fig:1d}), $\mc{M}_{\mr{circ}}$ (\ref{fig:circle}), $\mc{M}_{\mr{ladder}}$ (\ref{fig:ladder}), and $\mc{M}_{\mr{layer}}$ (\ref{fig:network}).
    All these Ising models are constructive context free. Intuitively, this is because their interaction graphs all have one distinguished dimension along which an elementary building block (that contains a constant number $k$ spins) is repeated ($i$ times) in a periodic fashion.
    In (\ref{fig:1d}) and (\ref{fig:circle}) there is only one dimension, in (\ref{fig:ladder}) and (\ref{fig:network}) the distinguished dimension is indicated as``direction of growth".
    This property is made precise in \cref{thm:main} \ref{thm:main:cf1}: constructive context free Ising models are uniquely characterised by  ${\sf E}_{\mc{M}}$ being regular and $\mc{M}$ being linear (or limited according to \cref{ssec:proofmain2} \ref{thm:main:cf1:iii}).}
            \label{fig:CCF}
\end{figure*}

\subsection{Ising model on a ladder graph}
\label{ssec:ladder}

Another class of Ising models with constructive context free languages is obtained by considering generalised ladder graphs as interaction graphs.  
For each such model the interaction graphs are given by a family of d-dimensional lattices, 
such that all lattices of the family have equal size along all but one dimension, i.e.\ increasing the system size amounts to adding spins along one distinguished dimension. 
It follows from \cref{thm:main} \ref{thm:main:cf1} that all these Ising models are constructive context free. 

To illustrate this, we consider $2$-dimensional ladders with constant width $k$ (\cref{fig:ladder}). The corresponding Ising model $\mc{M}_{\mr{ladder}}$ is defined by 
\begin{align}
    \begin{aligned}
        N_{\mr{ladder}}&\coloneqq  \{ i  k \mid i \geq 2 \} \\
        (E_{\mr{ladder}})_{i  k} &\coloneqq  (E_{\mr{ladder}})_{i  k}^{\mr{ver}} \cup (E_{\mr{ladder}})_{i  k}^{\mr{hor}}\\  
    \end{aligned}
\end{align}    
where
\begin{multline}
    (E_{\mr{ladder}})_{i  k}^{\mr{ver}} \coloneqq  \{(j  k+l,j  k+l+1) \mid \\
    0\leq j \leq i-1, \ 1 \leq l \leq k-1\}   
\end{multline}
contains the vertical edges and 
\begin{multline}
    (E_{\mr{ladder}})_{i  k}^{\mr{hor}} \coloneqq \{(j  k+l,(j+1)  k+l) \mid \\
    0 \leq j \leq i-2, \  1 \leq l \leq k \}    
\end{multline}
contains the horizontal edges.

$\mc{M}_{\mr{ladder}}$ is linear, as $\vert (E_{\mathrm{ladder}})_{n=i  k} \vert= 2i  k -i-k < 2  n$. 
Moreover, its edge language ${\sf E}_{\mr{ladder}}$ is regular, since
it can be written as a finite union of regular expressions
\begin{align}
    \begin{aligned}
    {\sf E}_{\mr{ladder}} &=  {\sf E}_{\mr{ladder}}^{\mr{ver}} \cup {\sf E}_{\mr{ladder}}^{\mr{hor}}\\
    {\sf E}_{\mr{ladder}}^{\mr{ver}} & \coloneqq \bigcup_{1\leq l \leq k-1} \bigl ((0^k)^*0^k0^{l-1}110^{k-l-1}(0^k)^* \\
    & \hphantom{\coloneqq \bigcup_{1\leq l \leq k-1} \bigl (} \cup (0^k)^*0^{l-1}110^{k-l-1}0^k(0^k)^*  \bigr ) \\
    {\sf E}_{\mr{ladder}}^{\mr{hor}}  &\coloneqq  \bigcup_{1\leq l \leq k}   (0^k)^* 
    0^{l-1}10^{k-1} 1 0^{k-l}(0^k)^* 
    \end{aligned}
\end{align}
From \cref{thm:main} \ref{thm:main:cf1} it follows that ${\sf L}_{\mr{ladder}}$ is constructive context free.
In addition, $\mc{M}_{\mr{ladder}}$ is not finite, so by \cref{thm:main}\ref{thm:main:regular} ${\sf L}_{\mr{ladder}}$ is not regular.

\subsection{Ising model on  layerwise complete graph}
\label{ssec:network}

Also layerwise complete graphs, which are used in many neural network models, 
define a class of Ising models with constructive context free language. 
To see this, consider interaction graphs  composed of $i$ layers of $k$ vertices, 
such that there is no edge between vertices within the same layer, 
and any two vertices from neighbouring layers are connected (\cref{fig:network}).
The corresponding Ising model  $\mc{M}_{\mr{layer}}$ is defined as 
\begin{align}
\begin{aligned}
    N_{\mr{layer}}&\coloneqq \{ i  k \mid i \geq 2 \} \\
    (E_{\mr{layer}})_{i  k} &\coloneqq \{(j  k+l,(j+1)  k+r) \mid \\
    & \hphantom{\coloneqq} 0 \leq  j \leq i-2,\  1 \leq l \leq k, \  1 \leq  r  \leq k\} 
\end{aligned}
\end{align}
Since 
\begin{align}
\vert (E_{\mr{layer}})_{n=i  k} \vert = (i-1)  k^2 < k  n
\end{align}
$\mc{M}_{\mr{layer}}$ is linear. 
Regularity of ${\sf E}_{\mr{layer}}$ can be seen from
\begin{align}
    {\sf E}_{\mr{layer}} = \bigcup_{1 \leq l \leq k, \ 1 \leq r \leq k} (0^k)^* 0^{l-1}10^{k-l} 0^{r-1}10^{k-r} (0^k)^*
\end{align}
Using \cref{thm:main} \ref{thm:main:cf1} it follows that ${\sf L}_{\mr{layer}}$ is constructive context free.
As $\mc{M}_{\mr{layer}}$ is not finite, by  \cref{thm:main} \ref{thm:main:regular} we conclude that ${\sf L}_{\mr{layer}}$ is not regular.

\subsection{2d Ising model}
\label{ssec:2D}

Let us now show that the Ising model on $2$d square lattices has a constructive context sensitive language. 
We denote the 2d Ising model as $\mc{M}_{\mr{2d}}$, and define it as
\begin{align}
    \begin{aligned}
        N_{\mr{2d}} &\coloneqq \{n^2 \mid n\geq 2 \} \\
        (E_{\mr{2d}})_{n^2} & \coloneqq (E_{\mr{2d}})_{n^2}^{\mr{hor}} \cup (E_{\mr{2d}})_{n^2}^{\mr{ver}}
    \end{aligned}
\end{align}
The edge set of size $n^2$ is split in horizontal and vertical edges:
\begin{align}\label{def:2dEdgeSplit}
    \begin{aligned}
        (E_{\mr{2d}})_{n^2}^{\mr{hor}} & \coloneqq \{ (i,i+1) \mid 1 \leq i \leq n^2-1, \ i \notin n\mathbb{N}   \} \\
        (E_{\mr{2d}})_{n^2}^{\mr{ver}} & \coloneqq  \{(i, i+n) \mid 1 \leq i \leq n^2-n \} 
    \end{aligned}
\end{align}
Its family of interaction graphs can be seen in \cref{fig:ladder}, 
with the difference that for $\mc{M}_{\mr{2d}}$, when increasing the system size both dimensions are scaled up simultaneously.

Using \cref{thm:main}\ref{thm:main:ccs}  we now prove that ${\sf L}_{\mr{2d}}$ is constructive context sensitive  by showing that its edge language ${\sf E}_{\mr{2d}}$ is context sensitive.
To this end, we build a  LBA that decides ${\sf E}_{\mr{2d}}$. 
Asserting that the input $w_1 \ldots w_m$ is well-formed, 
i.e.\ of the form $0^*10^*10^*$, can be achieved by a FSA, which can be simulated by the LBA. 
We can thus w.l.o.g.\ assume that the input is well-formed.
Next the LBA checks if $m=n^2$ for some natural number $n$. This is done by iterating over natural numbers $n$, starting with $n=1$.
In each step of the iteration the LBA computes $n^2$ and checks if this matches the length of the input $m$. If $n^2=m$ this subroutine terminates, if $n^2<m$ the LBA moves on with the next natural number $n+1$, 
and if $n^2>m$ the LBA rejects the input, as then $m$ is not a square number, i.e.\ not in $N_{\mr{2d}}$.
Explicitly, we use $(n+1)^2=n^2+2n+1$ to compute the square numbers. 
The LBA uses an additional tape $T_{n}$ to store $n$ in unary.  
Initially $n=1$ and the head of the LBA is placed over the first cell of the input tape. 
The LBA enters a loop: 
The head moves $2n+1$ cells to the right on the input tape. Note that this places the head over cell $(n+1)^2$.
The LBA now checks if the current cell is empty.
If yes, then $m\notin N_{\mr{2d}}$ and the LBA rejects the input.
If no, the LBA checks if the next cell is empty.
If no, it increases $n$ by one in the additional tape, and starts again. 
If yes, it accepts. 

Now the LBA traverses the input until its head reaches the first $1$.
While doing so it uses another additional tape $T_{f}$ to count the position of that first $1$ in the input string, $f$.
Then the LBA counts the number of $0$s between the first and the second $1$, $z$, and stores it in another additional tape  $T_{z}$.
Finally, it accepts the input 
if either $z=n-1$, corresponding to a vertical edge, 
or if $z=0$ and there exists no $k$ satisfying $f=k   n$ (this can be done since $n$ is written on $T_{n}$), 
corresponding to a horizontal edge.  

Note that, as 
\begin{equation}
\vert (E_{\mr{2d}})_{n^2} \vert = 2n(n-1) < 2n^2
\end{equation}
$\mc{M}_{\mr{2d}}$ is linear. However, for any $n\in \mathbb{N}$, 
\begin{equation}
(n,2n)\in (E_{\mr{2d}})_{n^2} 
\end{equation}
and hence $\mc{M}_{\mr{2d}}$ is not limited.  
So by \cref{cor:limited} ${\sf L}_{\mr{2d}}$ is not constructive context free.

We now show that ${\sf L}_{\mr{2d}}$ is not context free by using the pumping lemma for context free languages \cite{Ko97}. 
Assume that ${\sf L}_{\mr{2d}}$ was context free and let $p$ be the pumping length of ${\sf L}_{\mr{2d}}$.  Now consider
\begin{align}
l \coloneqq 0^{p^2}\bullet -^{2 p (p-1)} 
\end{align}
Note that $l \in {\sf L}_{\mr{2d}}$, as a configuration of $p^2$ spins has $2   p   (p-1)$ edges. 
When writing $l=uvwxy$, $v$ must be a non-empty string of $0$ symbols (let $\vert v \vert =:  k$),  while $x$ must be a non-empty string of $-$ symbols. 
Otherwise, pumping up $l$ would yield a mismatch between spin configuration and energy. 
Since $\vert vwx \vert \leq p$ we also have that $k \leq p$, so 
$uv^2wx^2y$ yields a  configuration with $p^2 + k$ spins. 
But $p^2 < p^2 +k < (p+1)^2$, so $p^2+k \notin N_{\mr{2d}}$ and thus $uv^2wx^2y \notin {\sf L}_{\mr{2d}}$. Hence, ${\sf L}_{\mr{2d}}$ is not context free.

\subsection{All-to-all Ising model}
\label{ssec:all2all}

Also the all-to-all Ising model (\cref{fig:all2all}), i.e.\ the Ising model with complete interaction graphs, 
has a constructive context sensitive language. We denote this model as $\mc{M}_{\mr{all}}$, and define it by 
\begin{align}
    \begin{aligned}
        N_{\mr{all}} & \coloneqq \{ n \mid n \geq 2 \} \\
        (E_{\mr{all}})_n & \coloneqq \{ (i,j) \mid 1 \leq i < j \leq n\}
    \end{aligned}
\end{align}  
Its edge language ${\sf E}_{\mr{all}}$ is regular, since 
\begin{equation}
{\sf E}_{\mr{all}}= 0^*10^*10^*
\end{equation}
Thus ${\sf E}_{\mr{all}}$ is in particular context sensitive and by \cref{thm:main} \ref{thm:main:ccs}, ${\sf L}_{\mr{all}}$
is constructive context sensitive. 
Since
\begin{align}
\vert (E_{\mr{all}})_n \vert = \frac{1}{2} n (n-1)
\end{align}
$\mc{M}_{\mr{all}}$ is not linear and, by \cref{thm:main} \ref{thm:main:cf1}, 
${\sf L}_{\mr{all}}$ is not constructive context free. 
In fact, by the pumping lemma \cite{Ko97}, ${\sf L}_{\mr{all}}$ can be proven to be not  context free. Assume ${\sf L}_{\mr{all}}$ was context free and denote its pumping length by $p$. 
Now take 
\begin{align}
    l \coloneqq 0^{p}\bullet -^{\frac{p (p-1)}{2}} \in {\sf L}_{\mr{all}}
\end{align}
Writing  $l = uvwxy$ as required by the pumping lemma, $v$ must be a non-empty string of $0$s and $x$ must a non-empty string of $-$ symbols.
Pumping up once yields $k\coloneqq \vert v \vert $ new spin symbols and thus increases the overall energy by 
\begin{equation}
e \coloneqq k  p+\frac{1}{2}  k    (k-1)
\end{equation}
Hence it must be the case that $x=-^e$. 
Pumping up a second time additionally adds
\begin{equation}
k   (p+k)+ \frac{1}{2}  k   (k-1) = e+k^2
\end{equation}  
more pair interactions but only $e$ more $-$ symbols. Thus, there is a mismatch between spin configuration and energy, and hence $uv^3wx^3y \notin {\sf L}_{\mr{all}}$. Therefore, ${\sf L}_{\mr{all}}$ is not context free. 

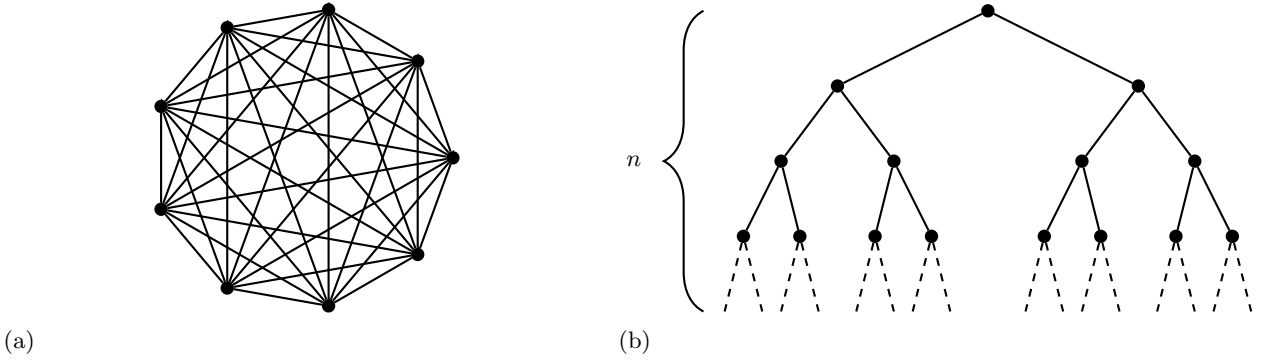
\begin{figure*}[t]
    \centering
    \begin{subfigure}[b]{.45\textwidth}
        \centering
            \begin{tikzpicture}[->,>=stealth',auto,
                thick,main node/.style={circle,draw,fill=black,inner sep=1.5pt}]{
              
              \foreach \a in {1,2,...,9}{
              \node[main node] (\a) at (\a*360/9: 2cm) {};
              \foreach \b in {1,...,\a}{
              \draw [-] (\a) -- (\b);
              }
              }
              }
              \end{tikzpicture}
            \caption{}\label{fig:all2all}
    \end{subfigure}
    \begin{subfigure}[b]{.45\textwidth}
        \centering
        \begin{tikzpicture}[ ->,>=stealth',auto,
            thick,main node/.style={circle,draw,fill=black,inner sep=1.5pt}]{
          \node[main node] (1) at (0,0) {};
          \node[main node] (2) at (-2,-1) {};
          \node[main node] (3) at (2,-1) {};
          \node[main node] (4) at (-2.75,-2) {};
          \node[main node] (5) at (-1.25,-2) {};
          \node[main node] (6) at (1.25,-2) {};
          \node[main node] (7) at (2.75,-2) {};
          \node[main node] (8) at (-3.25,-3) {};
          \node[main node] (9) at (-2.5,-3) {};
          \node[main node] (10) at (-1.5,-3) {};
          \node[main node] (11) at (-0.75,-3) {};
          \node[main node] (12) at (0.75,-3) {};
          \node[main node] (13) at (1.5,-3) {};
          \node[main node] (14) at (2.5,-3) {};
          \node[main node] (15) at (3.25,-3) {};

          \coordinate (16) at (-3.5,-4);
          \coordinate (17) at (-3,-4);
          \coordinate (18) at (-2.75,-4);
          \coordinate (19) at (-2.25,-4);
          \coordinate (20) at (-1.75,-4);
          \coordinate (21) at (-1.25,-4);
          \coordinate (22) at (-1,-4);
          \coordinate (23) at (-0.5,-4);
          \coordinate (24) at (0.5,-4);
          \coordinate (25) at (1,-4);
          \coordinate (26) at (1.25,-4);
          \coordinate (27) at (1.75,-4);
          \coordinate (28) at (2.25,-4);
          \coordinate (29) at (2.75,-4);
          \coordinate (30) at (3,-4);
          \coordinate (31) at (3.5,-4);

          \draw[-] (1) -- (2);
          \draw[-] (1) -- (3);
          \draw[-] (2) -- (4);
          \draw[-] (2) -- (5);
          \draw[-] (3) -- (6);
          \draw[-] (3) -- (7);
          \draw[-] (4) -- (8);
          \draw[-] (4) -- (9);
          \draw[-] (5) -- (10);
          \draw[-] (5) -- (11);
          \draw[-] (6) -- (12);
          \draw[-] (6) -- (13);
          \draw[-] (7) -- (14);
          \draw[-] (7) -- (15);

          \draw[-, dashed] (8) -- (16);
          \draw[-, dashed] (8) -- (17);
          \draw[-, dashed] (9) -- (18);
          \draw[-, dashed] (9) -- (19);

          \draw[-, dashed] (10) -- (20);
          \draw[-, dashed] (10) -- (21);
          \draw[-, dashed] (11) -- (22);
          \draw[-, dashed] (11) -- (23);

          \draw[-, dashed] (12) -- (24);
          \draw[-, dashed] (12) -- (25);
          \draw[-, dashed] (13) -- (26);
          \draw[-, dashed] (13) -- (27);

          \draw[-, dashed] (14) -- (28);
          \draw[-, dashed] (14) -- (29);
          \draw[-, dashed] (15) -- (30);
          \draw[-, dashed] (15) -- (31);
          \draw[-, decorate,decoration={brace,amplitude=15pt},xshift=-8pt,yshift=0pt]
              (-3.5,-4) -- (-3.5,0) node [black,midway, xshift=-20pt] {$n$};
          }
          \end{tikzpicture}
            \caption{}\label{fig:tree}
    \end{subfigure}
\caption{Interaction graphs of $\mc{M}_{\mr{all}}$ (\ref{fig:all2all}) and $\mc{M}_{\mr{tree}}$ (\ref{fig:tree}).
(\ref{fig:tree}) shows a perfect binary tree that contains $2^n-1$ vertices, and thus consists of $n$ individual levels. 
Increasing the system size in $\mc{M}_{\mr{tree}}$ adds entire levels to the tree. 
These Ising models, as well as  $\mc{M}_{\mr{2d}}$, have constructive context sensitive language. 
$\mc{M}_{\mr{all}}$ has regular edge language, but fails to be constructive context free as it is not linear.
In contrast, $\mc{M}_{\mr{tree}}$ and $\mc{M}_{\mr{2d}}$ have context sensitive edge languages and are linear. 
\cref{thm:main} \ref{thm:main:ccs} states that the crucial property for ${\sf L}_{\mc{M}}$ to be constructive context sensitive is that  ${\sf E}_{\mc{M}}$ be context sensitive.
}
\end{figure*}

\subsection{Ising model on perfect binary trees}
\label{ssec:tree}

Next we consider the Ising model $\mc{M}_{\mr{tree}}$ that uses perfect binary trees as interaction graphs (\cref{fig:tree}). This model is defined by 
\begin{align}
    \begin{aligned}
        N_{\mr{tree}} & \coloneqq \{ 2^n-1 \mid n \geq 2 \} \\
        (E_{\mr{tree}})_{2^n-1} & \coloneqq (E_{\mr{tree}})_{2^n-1}^{\mr{left}} \cup (E_{\mr{tree}})_{2^n-1}^{\mr{right}} 
    \end{aligned}
\end{align}
where the edge set of size $2^n-1$ is split into those that connect the parent vertex to its left child vertex and those that connect the parent vertex to its right child vertex:
\begin{align}\label{eq:edge-tree}
    \begin{aligned}
        (E_{\mr{tree}})_{2^n-1}^{\mr{left}} & \coloneqq \{ (i,2i) \mid 1 \leq  i \leq 2^{n-1}-1\}\\
        (E_{\mr{tree}})_{2^n-1}^{\mr{right}} & \coloneqq \{ (i, 2i+1) \mid 1 \leq i \leq 2^{n-1}-1 \} 
    \end{aligned}
\end{align} 
In order to apply \cref{thm:main} \ref{thm:main:ccs} we need to prove that ${\sf E}_{\mr{tree}}$ is context sensitive. To this end, consider the language that only encodes the system sizes $N_{\mr{tree}}$, 
\begin{align}
    {\sf N}_{\mr{tree}} \coloneqq \{ w_1 \ldots w_{2^n-1} \mid w_i \in \{0,1\}, \  n \geq 2 \} 
\end{align}
We now show that this language is context sensitive by constructing a LBA that decides it. Given an input string 
\begin{equation}
w_1 \ldots w_m \in \{0,1\}^*
\end{equation} 
the LBA checks if $m=2^n-1$ for some natural number $n$.
It does so by using an additional tape $T_{2^n}$ to store $2^n$ in unary. 
It starts with $n=1$ (so that $2^n=2$), and the head placed on the first cell of the input tape. 
Then it enters the following loop.  
It moves the head on the input tape $2^n$ cells to the right (this is possible because $2^n$ is stored on the additional tape).  
If the current cell is the last non-empty cell, it accepts. 
If the cell is empty, it rejects. 
Else, it doubles the number of symbols on the additional tape (so that it now contains $2^{n+1}$), moves its head back to the beginning of the input tape, and continues with the first step of the loop. 
This shows that ${\sf N}_{\mr{tree}}$ is context sensitive.

Next, note that the edge language is given by 
\begin{align} \label{eq:Etree}
     {\sf E}_{\mr{tree}}= \bigl ( {\sf E}_{\mr{tree}}^{\mr{left}} \cup {\sf E}_{\mr{tree}}^{\mr{right}} \bigr ) \cap {\sf N}_{\mr{tree}}
\end{align}
where
\begin{align}
    \begin{aligned}
        {\sf E}_{\mr{tree}}^{\mr{left}} &\coloneqq \{ 0^{i-1} 1 0^{i-1} 1 0^* \mid i \geq 1 \}\\
        {\sf E}_{\mr{tree}}^{\mr{right}} &\coloneqq \{ 0^{i-1}1  0^{i} 1 0^* \mid i \geq 1 \}
    \end{aligned}
\end{align}
Both ${\sf E}_{\mr{tree}}^{\mr{left}}$ and ${\sf E}_{\mr{tree}}^{\mr{right}}$ are context free, 
as can be seen by constructing two PDAs $P_{\mr{left}}$ and $P_{\mr{right}}$ that accept these two languages, respectively. 
(This can also directly be seen from the fact that both languages are essentially of the form $\{a^nb^n \mid n\geq 1\}$.)
$P_{\mr{left}}$ uses its stack to count the number of zeros in front of the first $1$, and then it compares this number against the number of zeros in front of the second $1$. 
If the two numbers coincide and the string contains no further 1, it accepts, else it rejects.
$P_{\mr{right}}$ does the same, except for ignoring the first symbol after the first $1$ if it is $0$ and rejecting if it is $1$. 

Finally, from \eqref{eq:Etree} and the closure properties of context sensitive languages \cite{Ko97},  it follows that ${\sf E}_{\mr{tree}}$  is context sensitive.
Hence, by  \cref{thm:main} \ref{thm:main:ccs}, ${\sf L}_{\mr{tree}}$ is constructive context sensitive.

$\mc{M}_{\mr{tree}}$ is not limited, 
as for any $n\geq 2$, the edge 
\begin{equation}
(2^{n-1}-1,2^{n}-2) \in (E_{\mr{tree}})_{2^n-1}^{\mr{left}}
\end{equation}
 is long-range.  
Hence, by \cref{cor:limited},   ${\sf L}_{\mr{tree}}$ is not constructive context free.

Moreover, ${\sf L}_{\mr{tree}}$ is not context free. This can be proven with the pumping lemma of context free languages \cite{Ko97}. Assume ${\sf L}_{\mr{tree}}$ was context free and let $p$ be its pumping length.
Take $n$ to be the smallest natural number that satisfies $2^n-1 \geq p$. Consider 
\begin{align}
    l = 0^{2^{n}-1} \bullet -^{2^{n}-2} \in {\sf L}_{\mathrm{tree}}
\end{align}
Writing $l=uvwxy$, $v$ must be a non-empty string of $0$s and $x$ a non-empty string of $-$ symbols. 
Then pumping up once yields a string that corresponds to configuration of $2^{n}-1 + k$ spins, where $k \coloneqq \vert v \vert$. As $k\leq p \leq 2^{n}-1$ it follows that
$2^{n}-1+k < 2^{n+1}-1$. Additionally using that  $k>0$ shows that 
$2^{n}-1+k \notin N_{\mathrm{tree}}$ and hence $uv^2wx^2y \notin {\sf L}_{\mathrm{tree}}$. Thus, ${\sf L}_{\mathrm{tree}}$ is not context free.

\section{Conclusions and Outlook}
\label{sec:conclusions}

In this work we have introduced a new complexity measure for Ising models and fully classified Ising models according to it (\cref{thm:main}).  
The complexity measure consists of classifying the decision problem corresponding to the function graph of the Hamiltonian of an Ising model in the Chomsky hierarchy. 

In order to establish this classification, 
we have identified certain properties of interaction graphs of Ising models. 
These properties can be divided into two classes: 
those that capture the complexity of interaction graphs (viz.\ the complexity of the edge language, \cref{def:EdgeLanguage}), 
versus those that capture the scaling of interaction graphs (viz.\ finite, limited and linear, \cref{def:PropIsing}). 

In our main result 
we have unveiled which properties of interaction graphs correspond to which complexity level of an Ising model in a one-to-one manner. 
We have then used the classification of \cref{thm:main} to compute the complexity of the 1d Ising model, the Ising model on ladder graphs, on layerwise complete graphs, 
the 2d Ising model, the all-to-all Ising model, and the Ising model on perfect binary trees (\cref{tab:examples}).  

Among other things, this work raises the question of how the complexity measure provided by classifying  ${\sf L}_{\mc{M}}$ in the Chomsky hierarchy differs from existing complexity measures for spin models, such as the computational complexity of GSE.
Specifically, these two complexity measures seem to have different easy-to-hard thresholds.
Considering Ising models, GSE can be proven to be easy (polytime computable) on planar and hard (NP-complete) on non-planar crystal lattices \cite{Is00}.
Thus, planarity seems to be the key property of Ising models that determines the hardness of GSE.
However, it seems to be unclear to what extent this also applies to Ising models which are not defined on regular lattices.
In contrast, our classification (\cref{thm:main})
reveals different properties, seemingly not related to planarity, that determine the complexity of ${\sf L}_{\mc{M}}$.
In addition, there exist non-planar Ising models (e.g.\ such on ladder graphs) with easy (constructive context free) ${\sf L}_{\mc{M}}$
as well as planar Ising models (e.g.\ the 2d Ising model or the Ising model on perfect binary trees) with hard (constructive context sensitive) ${\sf L}_{\mc{M}}$.
Both observations illustrate that the relation between the two complexity measures is to be further explored.

A different way of comparing the two measures consists of investigating the computational complexity of deciding ${\sf L}_{\mc{M}}$, that is, the time resources a Turing machine needs to decide  ${\sf L}_{\mc{M}}$---this is done in \cite{St21} for general spin models. 
Conversely, one could classify the language of GSE in the Chomsky hierarchy, and thereby unveil the grammar (i.e.\ local structure) of the set of yes instances of the ground state energy problem. 

To what extent can the complexity measure provided by ${\sf L}_{\mc{M}}$ as well as its characterisation be extended beyond homogeneous Ising models.
We expect that a characterisation similar to \cref{thm:main} can be derived for non-homogeneous Ising models, by modifying the edge language such that in addition to the positions of the two interacting spins, each string of it also contains the coupling of the respective interaction e.g.\ via $00\alpha  0 \alpha 0$ instead of $001010$ to encode that the interaction between spins $3$ and $5$ has coupling $\alpha$.
We further expect that ${\sf L}_{\mc{M}}$ can be defined similarly for general spin models instead of Ising models.
We however do not know if these generalisations allow for a similar characterisation of ${\sf L}_{\mc{M}}$.
Considering extensions beyond classical spin models, e.g.\ quantum spin models, we expect that already the first step, i.e.\ the encoding in terms of formal languages might prove difficult. 

It could be also interesting to attempt a similar complexity classification of Ising models based on graph languages and graph grammars instead of their string based counterparts.
Encoding the Hamiltonian of an Ising model as a formal languages enforces a total order of the spins. This could be avoided by using graph languages and graph grammars.
A graph language is a set of graphs 
and graph grammars generalise the production rule of string grammars to operate directly on graphs \cite{Ro97}. 
While encoding Ising models as graph languages thus seems more natural, graph grammars lack the well-studied complexity hierarchy of string grammars. 
Hence, it is not clear how encoding Ising models as graph languages could give rise to a complexity measure.

From a broader perspective, this work ---together with \cite{St21}--- establishes a new connection between spin models and theoretical computer science. 
Among other reasons this connection is motivated by the recent discovery that certain spin models such as the 2d Ising model with fields or the 3d Ising model are universal, i.e.\ can simulate arbitrary other spin models \cite{De16b,De20d}. 
The complexity measure for Ising models introduced in this work might allow for a new characterisation of universal spin models, possibly in terms of the complexity of their languages. 
From a more conceptual perspective, connecting spin models and theoretical computer science might reveal if universal spin models are, in some way, related to universal Turing machines, i.e.\ Turing machines that can simulate the computation of any other Turing machine.

Finally, this work can be seen as a first step in the characterisation of physical systems in terms of grammars.
We think that this approach is meaningful within a much broader context.
Consider for instance the time evolution of a discrete, dynamical system. 
Encoding the transition from configuration $c_1$ (at time $t$) to configuration $c_2$ (at time $t+1$) in terms of a grammar that allows one to derive $c_2$ from $c_1$, the grammar can be seen as capturing the dynamics or the physical interactions of the system.
It seems plausible that the grammar then also encodes crucial properties of the system which might be revealed by using tools from formal language theory.  
We thus also see this work as an invitation to studying general physical systems or processes in the light of grammars.

\section*{Acknowledgements} 
We thank Sebastian Stengele for many discussions, as well as everyone in the group---Tom\'a\v{s} Gonda, Andreas Klingler and Mirte van der Eyden. 
We also thank David B\"ansch for joint work on graph grammars and graph languages. 
We thank Thomas M\"uller,  Hadil Karawani and Sahra Styger from the University of Konstanz for many  discussions about counterfactuals. 
We  acknowledge funding from Austrian Science Fund (FWF) via the FWF START Prize Y1261-N.

\appendix 
\section{Proof of \cref{thm:main}}
\label{sec:proofmain}

Here we prove \cref{thm:main} with one subsection for each statement. 

\subsection{Proof of \cref{thm:main}\ref{thm:main:regular}}
\label{ssec:proofmain1}

If $N_{\mc{M}}$ is finite then so is ${\sf L}_{\mc{M}}$. Thus, ${\sf L}_{\mc{M}}$ is trivially regular \cite{Ko97}. 

In order to prove the ``only if" direction using the pumping lemma for regular languages \cite{Ko97}, 
we prove that an Ising model with infinite $N_{\mc{M}}$ cannot be regular. 
To this end, assume such a ${\sf L}_{\mc{M}}$ was regular and let $p$ be its pumping length. 
As $\mc{M}$ is infinite, there exists a configuration of length $q>p$. Hence, the string
\begin{equation}
0^q\bullet -^e \in {\sf L}_{\mc{M}}
\end{equation}
with $e \coloneqq \vert (E_{\mc{M}})_q \vert $ is contained in ${\sf L}_{\mc{M}}$. 
Note that $e \leq \frac{1}{2}q(q-1)$, and hence, for any $n>\frac{1}{2}q(q-1)+1$ we have $\vert E_n \vert > e$. 
Pumping up $k$ times yields a string of the form $0^{q+kp}\bullet -^e$.
Now choosing $k$ large enough such that 
\begin{align}
    q+kp>\frac{1}{2}q(q-1)+1
\end{align}
this string is not contained in  ${\sf L}_{\mc{M}}$, since $\vert E_{q+kp} \vert > e$, so $H_{\mc{M}}(0^{q+kp})\neq  -e$.
Thus, an infinite Ising model cannot have a regular language.

\subsection{Proof of \cref{thm:main}\ref{thm:main:cf1}}
\label{ssec:proofmain2}

In order to prove that 
\begin{quote}
${\sf L}_{\mc{M}}$ is constructive context free $\iff  \mc{M} $ is linear and ${\sf E}_{\mc{M}} $ is regular
\end{quote}
 we prove the following four statements: 
\begin{enumerate}[label=(\alph*)]

\item \label{thm:main:cf1:i}
${\sf L}_{\mc{M}}$ is context free $ \implies \mc{M} $ is linear 

\item \label{thm:main:cf1:ii}
${\sf L}_{\mc{M}}$ is constructive context free $ \implies  {\sf E}_{\mc{M}} $ is regular

\item \label{thm:main:cf1:iii}
$ \mc{M} $ is linear and ${\sf E}_{\mc{M}} $ is regular $ \implies  \mc{M} $ is limited

\item \label{thm:main:cf1:iv}
$ \mc{M} $ is limited and ${\sf E}_{\mc{M}} $ is regular   $  \implies {\sf L}_{\mc{M}}$ is constructive context free  
\end{enumerate}

Combining \ref{thm:main:cf1:i} and  \ref{thm:main:cf1:ii} (and the fact that constructive context free is included in context free) yields the forward direction of the statement,
and combining \ref{thm:main:cf1:iii} and \ref{thm:main:cf1:iv} yields the other direction.  

Let us now prove each of the statements.

\subsubsection*{\ref{thm:main:cf1:i}  If ${\sf L}_\mc{M}$ is  context free, then $\mc{M}$ is linear}

As by assumption ${\sf L}_{\mc{M}}$ is context free, 
we claim that so is the language containing the configuration of minimal energy for each system size, 
\begin{equation}
({\sf L}_{\mc{M}})_{\mr{min}} \coloneqq \{ 0^n -^{e_n} \mid n \in N_{\mc{M}} \}
\end{equation}
where $e_n \coloneqq \vert (E_{\mc{M}})_n \vert $. 
This holds since $({\sf L}_{\mc{M}})_{\mr{min}}$ can be obtained from ${\sf L}_{\mc{M}}$ by first intersecting with the regular language $0^*\bullet -^*$ and then applying the homomorphism that maps $\bullet $
to the empty string and acts as identity on $\{0,1,+,-\}$. 
Since the class of context free languages is closed both with respect to intersections with regular languages and homomorphisms \cite{Ko97}, this proves the claim that  $({\sf L}_{\mc{M}})_{\mr{min}}$ is context free.

As $({\sf L}_{\mc{M}})_{\mr{min}}$ is context free, its image under the Parikh map  
\begin{equation}
P(0^n -^{e_n})= (n, e_n)
\end{equation}
 is a semilinear subset of $\mathbb{N}^2$, i.e.\ a union of finitely many linear subsets $U_1,  \ldots , U_r \subseteq \mathbb{N}^2 $ \cite{Ko97}.
We now construct a natural number $k$ such that for all $n \in N_{\mc{M}}$, $e_n \leq k n$.
Take any $(n, e_n)$ from the image of $P$.
Then there is an  $i \leq r$ such that  $(n,e_n) \in U_i$. 
As $U_i$ is linear, there exist $u_0 \in \mathbb{N}^2, u_1, \ldots ,u_d\in \mathbb{N}^2\setminus \{(0,0)\}$,
such that any element in $U_i$ can be written as 
\begin{equation}
u_0+\lambda_1u_1+ \ldots +\lambda_du_d
\end{equation} 
with $\lambda_j$ natural numbers. Thus, denoting $u_j= (v_j, w_j)$ we in particular have
\begin{align}
    \frac{e_n}{n} = \frac{w_0+ \lambda_1w_1+ \ldots +\lambda_d w_d}{v_0+ \lambda_1v_1+ \ldots +\lambda_dv_d}
\end{align}
Now note that for any $u_j$ it holds that $v_j$ is strictly positive. 
For assume that $v_j=0$. Then, by the linearity of $U_i$, 
\begin{equation}
0^{n+\lambda   0} \bullet -^{e_n+\lambda w_j}\in {\sf L}_{\mc{M}}
\end{equation} 
so a single spin configuration, $0^n$, would have energies $-^{e_n}$ and  $-^{e_n+\lambda w_j}$. In other words, the relation between spin configuration and energy would no longer be functional.
Moreover, $v_0$ cannot be zero either, by  \cref{def:HamLanguage}.  

To finish the proof, take 
\begin{align}
    k_i \coloneqq \max \Bigl\{ \frac{w_l}{v_j} \mid l,j \leq d \Bigr\}  
\end{align}
Then it is easy to see that
\begin{align}
\frac{e_n}{n} \leq k_i
\end{align}
Defining $k$ to be the maximum taken over $\{k_i \mid i \leq r \}$ shows that for any $n\in N_{\mc{M}}$, $\vert (E_{\mc{M}})_n \vert  \leq k   n$
and hence proves the claim.

\subsubsection*{\ref{thm:main:cf1:ii} If ${\sf L}_{\mc{M}}$ is constructive context free, then ${\sf E}_{\mc{M}}$ is regular}

As ${\sf L}_{\mc{M}}$ is constructive context free there exists a constructive PDA $P$ that accepts ${\sf L}_{\mc{M}}$.
We prove the claim by first using $P$ to construct a second PDA $P_C$ that decides ${\sf E}_{\mc{M}}$, and showing that there exists a finite bound on the stack memory of $P_C$. 
Since a finite stack can be simulated by a FSA (by increasing the number of states),  $P_C$ can be transformed into a FSA, which proves the claim.

So let us consider a potential edge
\begin{equation}
    \langle i, j-1-1, n-j \rangle \coloneqq 0^{i-1}10^{j-i-1}10^{n-j} 
\end{equation}
In order to decide if it is in $ {\sf E}_{\mc{M}}$, $P_C$ computes $C_{\mc{M}}(n,i,j)$ defined in \cref{eqn:corr}, 
by simulating $P$'s computation on  the four input spin configurations 
\begin{equation}
0^{i-1}10^{j-i-1}10^{n-j}, \quad 
0^{i-1}10^{n-i}, \quad 
0^{j-1}10^{n-j}, \quad 
0^n 
\end{equation}
and summing the four energies appropriately. 

We now prove that, since $P$ is constructive, $P_C$ can be taken to be a FSA.
Let $(I_m)_{m=1,\ldots , r_n}$ be the unique partition of $(E_{\mc{M}})_n$ that witnesses that $P$ is constructive. 
At step $m$ of the main iteration of $P$ (cf.\ \cref{def:constructive} \ref{def:constructive:i}), $P$ computes the energy contribution from interactions contained in $I_m$ and adds it to its stack. Consequently, $P_C$ computes 
\begin{align}\label{eqn:corrI}
    \begin{aligned}
    C_{\mc{M}}( n, i,j)_{I_m} \coloneqq &-\frac{1}{4} \bigl[H_{\mc{M}}\vert_{I_m}(0^n) \\
    &+ H_{\mc{M}}\vert_{I_m}(0^{i-1}10^{j-i-1}10^{n-j})\\
    &- H_{\mc{M}}\vert_{I_m}(0^{i-1}10^{n-i})\\
    &- H_{\mc{M}}\vert_{I_m}(0^{j-1}10^{n-j})\bigr ]
    \end{aligned}
\end{align}
and adds the result to its stack.
By \cref{def:constructive} the energy that each $I_m$ contributes is upper bounded by the number of states of $P$, and so in particular it is finite.  
Thus, summing up the four terms of 
\cref{eqn:corrI} can be done in the states of $P_C$, 
and we can assume that $P_C$ only uses its stack to accumulate 
\begin{equation}\label{eq:sumCM}
\sum_{m = 1}^{r_n} C_{\mc{M}}( n, i,j)_{I_m}
\end{equation} 
By construction $C_{\mc{M}}( n, i,j)_{I_m}$ is $+1$ if $(i,j) \in I_m$ and $0$ else.
Thus, a finite stack suffices to compute \eqref{eq:sumCM}, and hence this can be done in the states of $P_C$. 
This makes $P_C$ a FSA.  

Finally, if $n \notin  N_{\mc{M}}$, $P_C$ rejects by construction, as so does $P$. If $n\in N_{\mc{M}}$, $P_C$ accepts the input if and only if $C_{\mc{M}}( n, i,j) = 1$.
So $P_C$ correctly decides ${\sf E}_{\mc{M}}$, which proves that ${\sf E}_{\mc{M}}$ is regular.

\subsubsection*{\ref{thm:main:cf1:iii}  If $\mc{M}$ is   linear  and ${\sf E}_{\mc{M}}$ is regular, then $\mc{M}$ is limited}

We prove that if ${\sf E}_{\mc{M}}$ is regular and $\mc{M}$ is not limited, then $\mc{M}$ cannot be linear.
To this end,  
for any natural number $k$, assuming ${\sf E}_{\mc{M}}$ is regular,  we construct a natural number $l$ such that $\mc{M}$ contains more than $k  l$ edges of length $l$, i.e.\ it is not linear.

By assumption ${\sf E}_{\mc{M}}$ is regular. Let $F$ be a FSA that accepts it and denote the number of states of $F$ by $b$.
Consider an edge $\langle p,q,r \rangle \in {\sf E}_{\mc{M}}$ with $p > b$.
When accepting $\langle p,q,r \rangle $ there has to be at least one state that $F$ enters twice before reaching the first $1$ with its head.
Thus, $F$ contains a loop in its transition rules. 
Denote the number of transitions that are contained in this loop by $w_p$. 
Then, for any natural number $n_p$, 
\begin{equation}
\langle p+n_p   w_p, q,r \rangle \in {\sf E}_{\mc{M}}
\end{equation} 
By a similar reasoning, for an edge $\langle p,q,r \rangle$ with $q>b$, 
 $F$ must enter a loop after reading the first $1$ and before reading the second $1$. 
Denote the length of the corresponding loop by $w_q$. Then for any natural number $n_q$, 
\begin{equation} 
\langle p, q+n_q   w_q,r \rangle \in {\sf E}_{\mc{M}}
\end{equation} 
Similarly, for an edge $\langle p,q,r \rangle$ with $r>b$, $F$ enters a loop after reading the second $1$. Denote the number of transitions in this loop as $w_r$. 
Then for any natural number $n_r$, 
\begin{equation} 
\langle p,q, r+n_r   w_r \rangle \in {\sf E}_{\mc{M}}
\end{equation}

Now, since $\mc{M}$ is not limited, there exists an edge $\langle p,q,r \rangle \in {\sf E}_{\mc{M}}$ with $p,q,r > b$. 
By the above reasoning,
there exist natural numbers $w_p,w_q,w_r$ such that 
for any $n_p,n_q,n_r$ 
\begin{align}
     \langle p+n_p  w_p, q+n_q  w_q, r+n_r  w_r \rangle  \in {\sf E}_{\mc{M}}
\end{align}
Take any natural number $m$ and define 
\begin{align}
    l_m\coloneqq p+q+r+2+m  w_pw_qw_r 
\end{align}
We will now show that for any $k$, choosing $m$ appropriately, there are more than $k   l_m$ words of length $l_m$ and hence $\mc{M}$ is not linear.
To this end, take any $m_p,m_q,m_r$ that satisfy $m_p+m_q+m_r = m$. Then, the edge 
\begin{align}
    \langle p+m_p   w_pw_qw_r, q+ m_q   w_pw_qw_r, r + m_r   w_pw_qw_r\rangle
\end{align}
is contained in ${\sf E}_{\mc{M}}$ and has length  $l_m$.
Thus the number of edges of length $l_m$ is at least as big as the number of triples $(m_p,m_q,m_r)$ that sum to $m$,  
\begin{align}
\begin{aligned}
 & \bigl\vert \{(m_p,m_q,m_r) \mid m_p+m_q+m_r = m\} \bigr\vert \\
     &=\sum_{m_p=0}^{m} \sum_{m_q=0}^{m-m_p} 1 = \frac{1}{2}m^2+\frac{3}{2}m+1
     \end{aligned}
\end{align}
So, while $l_m$ grows linearly with $m$, the number of words of length $l_m$ grows at least quadratically with $m$. Thus,
for any $k\in \mathbb{N}$, choosing $m$
appropriately yields more than $k  l_m$ words of length $l_m$. Hence $\mc{M}$ is not linear.

\subsubsection*{\ref{thm:main:cf1:iv} If $\mc{M}$ is   limited and ${\sf E}_{\mc{M}}$ is regular, then ${\sf L}_{\mc{M}}$ is constructive context free}

We prove the claim by building a constructive PDA that accepts ${\sf L}_{\mc{M}}$.
Let $F$ be a FSA that accepts ${\sf E}_{\mc{M}}$ and let $b$ denote the number of states of $F$.
As a first step, we use $F$ to decompose ${\sf E}_{\mc{M}}$ into $8$ disjoint subsets, represented by eight finite sets \eqref{eq:fantastic8}. 
In the second step, for each of these sets, we build a PDA that computes the energy contribution corresponding to the edges in that set.  
Putting together these contributions shows that  ${\sf L}_{\mc{M}}$ can be recognised by a constructive PDA.

\medskip
\paragraph*{Decomposing ${\sf E}_{\mc{M}}$.} 
Take any edge $\langle p,q,r\rangle \in {\sf E}_{\mc{M}}$.
If $F$ enters a loop when processing $0^p$, 
we can w.l.o.g.\ assume that this loop is irreducible in the sense that it contains each state at most once; 
otherwise we decompose it until it is irreducible.
Denote the number of states of this loop by $w_p$. 
Then we can write $p = v_p + n_p   w_p$ for some natural number $n_p$. 
Note that 
\begin{equation}
\langle v_p+ n   w_p,q,r\rangle \in {\sf E}_{\mc{M}}
\end{equation} for any natural number $n$.
We call $(w_p,v_p)$ the $1$-loop-parameters of $\langle p,q,r\rangle$ ($1$- to indicate that the loop occurs in $p$ and not in $q$ or $r$) and 
say that $\langle p,q,r\rangle$ is $1$-periodic if there exist $1$-loop-parameters $(w_p,v_p)$ and a natural number $n_p$ such that $p = v_p + n_p   w_p$. 
Next, we define the set of $1$-loop-parameters that correspond to valid edges in ${\sf E}_{\mc{M}}$, requiring that all such $1$-loop-parameters describe irreducible loops, 
\begin{equation}
    P_{\mr{l}}\coloneqq \{ (w_p,v_p) \mid 
     (w_p,v_p) \text{ $1$-loop-parameters of ${\sf E}_{\mc{M}}$}\} 
\end{equation}
Note that $v_p,w_p \leq b$ as otherwise the loop would not be irreducible.
Thus, $P_{\mr{l}}$ is a finite set.

If $\langle p,q,r\rangle$ is not $1$-periodic, we say it is $1$-finite. We define 
\begin{multline}
    P_{\mr{f}}\coloneqq \{ p \mid \exists q,r \ \textrm{s.t.} \\
       \ \langle p,q,r \rangle  \in {\sf E}_{\mc{M}}  \text{ and } \ \nexists v_p : \ (p,v_p) \in P_{\mr{l}} \}
\end{multline}
Note that any $p \in P_{\mr{f}}$ must satisfy $p \leq b$, so $P_{\mr{f}}$ is also a finite set.
Note also that, by construction, any edge $\langle p,q,r\rangle  \in {\sf E}_{\mc{M}}$ is either $1$-finite or $1$-periodic, 
i.e.\ either $p \in P_{\mr{f}}$ or $p=w_p  n+v_p$ for a unique $\langle w_p,v_p \rangle  \in P_{\mr{l}}$ and $n \in \mathbb{N}$.

In exactly the same way we define $2$-periodicity, $2$-finiteness, $3$-periodicity and $3$-finiteness of an edge $\langle p,q,r\rangle$, where periodicity or finiteness refers to $q$ and $r$, respectively, as well as $2$-loop-parameters and $3$-loop-parameters and the corresponding sets 
\begin{align}
\begin{aligned}
    Q_{\mr{l}}\coloneqq &\{ (w_q,v_q) \mid  (w_q,v_q) \text{ $2$-loop-paramaters of ${\sf E}_{\mc{M}}$} \}\\
    Q_{\mr{f}}\coloneqq &\{ q \mid  \exists p,r \ \textrm{s.t.} \\
     &   \langle p,q,r \rangle \in E_{\mc{M}} \ \text{and} \ \nexists v_q: \ (q,v_q) \in Q_{\mr{l}} \}\\
    R_{\mr{l}}\coloneqq &\{ (w_r,v_r) \mid  (w_r,v_r) \text{ $3$-loop-paramaters of ${\sf E}_{\mc{M}}$}\}\\
    R_{\mr{f}}\coloneqq &\{ r \mid  \exists p,q \ \textrm{s.t.} \\ 
    &    \langle p,q,r \rangle \in E_{\mc{M}} \ \text{and} \ \nexists v_r: \ (r,v_r) \in R_{\mr{l}} \}
\end{aligned}
\end{align}
In addition, we define sets of combinations of $p,q,r$ that lead to valid edges in ${\sf E}_{\mc{M}}$,  
\begin{align}\label{eq:fantastic8}
    \begin{aligned}
    E_{\mr{fff}} \coloneqq &\{ (p,q,r) \in P_{\mr{f}}\times Q_{\mr{f}} \times R_{\mr{f}} \mid  \langle p,q,r \rangle  \in {\sf E}_{\mc{M}} \} \\
    E_{\mr{lff}}  \coloneqq &\{ ((w_p,v_p), q, r) \in P_{\mr{l}}\times Q_{\mr{f}} \times R_{\mr{f}} \mid
       \langle v_p,q,r \rangle  \in {\sf E}_{\mc{M}}\} \\
    E_{\mr{flf}}  \coloneqq &\{ (p,(w_q,v_q), r) \in P_{\mr{f}}\times Q_{\mr{l}} \times R_{\mr{f}}  \mid
       \langle p,v_q,r\rangle \in {\sf E}_{\mc{M}} \}\\
    E_{\mr{ffl}}  \coloneqq &\{ (p,q,(w_r,v_r)) \in P_{\mr{f}}\times Q_{\mr{f}} \times R_{\mr{l}}  \mid 
        \langle p,q,v_r\rangle \in {\sf E}_{\mc{M}} \}\\
    E_{\mr{llf}}  \coloneqq &\{ ((w_p,v_p),(w_q,v_q), r) \in P_{\mr{l}}\times Q_{\mr{l}} \times R_{\mr{f}}  \mid\\
    &   \langle v_p,v_q,r\rangle \in {\sf E}_{\mc{M}} \}\\
    E_{\mr{lfl}} \coloneqq &\{ ((w_p,v_p),q, (w_r,v_r)) \in P_{\mr{l}}\times Q_{\mr{f}} \times R_{\mr{l}}  \mid\\
    &   \langle v_p,q,v_r\rangle \in {\sf E}_{\mc{M}} \}\\
    E_{\mr{fll}}  \coloneqq &\{ (p,(w_q,v_q), (w_r,v_r)) \in P_{\mr{f}}\times Q_{\mr{l}} \times R_{\mr{l}}  \mid\\
    &   \langle p,v_q,v_r\rangle \in {\sf E}_{\mc{M}} \}\\
    E_{\mr{lll}}  \coloneqq &\{ ((w_p,v_p),(w_q,v_q), (w_r,v_r)) \in P_{\mr{l}}\times Q_{\mr{l}} \times R_{\mr{l}}  \mid \\
    &    \langle v_p,v_q,v_r\rangle \in {\sf E}_{\mc{M}} \}
    \end{aligned} 
\end{align}

Note that each of these sets is finite, and that they are all disjoint.  
Note also that if $\mc{M}$ is limited then $E_{\mr{lll}}$ is empty. 
So this decomposes ${\sf E}_{\mc{M}}$ into $7$ disjoint nonempty subsets. Explicitly, define their union
\begin{align}\label{eq:curlyE}
\mc{E}\coloneqq E_{\mr{fff}} \cup 
              E_{\mr{lff}} \cup 
              E_{\mr{flf}} \cup 
              E_{\mr{ffl}} \cup 
              E_{\mr{llf}} \cup 
              E_{\mr{lfl}} \cup 
              E_{\mr{fll}}
\end{align}
For each such set, any of its elements describes a subset $I_e\subseteq {\sf E}_{\mc{M}}$ of edges of $\mc{M}$.
For $(p,q,r) \in E_{\mr{fff}}$ this is a singleton  $I_{(p,q,r)} = \{\langle p,q,r \rangle \}$, 
but for elements of any set other than $E_{\mr{fff}}$, $I_e$ is infinite. 
For example, for $((w_p,v_p), (w_q,v_q),r) \in E_{\mr{llf}}$ we have 
\begin{equation}
I_{((w_p,v_p), (w_q,v_q),r)} = \{ \langle v_p+n  w_p,v_q+m  w_q,r \rangle \mid n,m \in \mathbb{N} \}
\end{equation}
 In other words, an edge  $\langle p^{\prime},q^{\prime},r^{\prime} \rangle $ is contained in $ I_{((w_p,v_p), (w_q,v_q),r)}$ if and only of 
\begin{align}
    \begin{aligned}
        p^{\prime} &= v_p \mod w_p\\
        q^{\prime} &= v_q \mod w_q \\
        r^{\prime} &= r
    \end{aligned}
\end{align}

By construction, any edge $\langle p^{\prime},q^{\prime},r^{\prime}\rangle \in {\sf E}_{\mc{M}}$ is described by a unique $e\in \mc{E}$.  
Thus, we obtain a partition of ${\sf E}_{\mc{M}}$ into a finite number of disjoint subsets:
\begin{align}
    {\sf E}_{\mc{M}} = \bigcup _{e \in \mc{E} } I_e 
\end{align}

\medskip
\paragraph*{Building the PDAs.} 
In order to build a constructive PDA $P$ that accepts ${\sf L}_{\mc{M}}$,
we build a constructive PDA $P_e$ for every $e \in \mc{E}$.
The idea is the following. 
First, since ${\sf E}_{\mc{M}}$ is regular, well-formedness of the input can easily be checked by a FSA, and thus also be simulated in the states of $P_e$.
So henceforth we shall assume that all inputs are well-formed, i.e.\ of the form 
\begin{equation}
s_1 \ldots s_n\bullet \sigma^k
\end{equation} 
Given a well-formed input, $P_e$ accumulates the energy contributions that correspond to edges in $I_e$ on its stack, that is, 
$P_e$ computes $H_{\mc{M}}\vert_{I_e}(s_1 \ldots s_n)$.
The required constructive PDA $P$ for ${\sf L}_{\mc{M}}$ is  obtained by running all $P_e$s in parallel while providing access to the same stack, so that  $P$ accumulates 
\begin{equation}\label{eq:energy-stack}
\sum_{e \in \mc{E}} H_{\mc{M}}\vert_{I_e}(s_1 \ldots s_n)
\end{equation} 
on its stack.
Since $\mc{E}$ is finite, there is  a finite number of PDAs $P_e$, and hence their parallel simulation can be performed by a PDA, $P$. 
Moreover, since $\bigcup _{e \in \mc{E} } I_e$ is a partition of ${\sf E}_{\mc{M}}$ into disjoint subsets, equation \eqref{eq:energy-stack} equals  $H_{\mc{M}}(s_1 \ldots s_n)$. 
Finally, $P$ compares its stack content to the input energy $\sigma^k$ and accepts if and only if the two values are equal.
All PDAs $P_e$ are built such that $P$ is constructive.

Let us construct the PDAs $P_e$ for any $e\in \mathcal{E}$. 
We will start by considering $e\in E_{\mr{fff}}$, and then continue with the following cases of \eqref{eq:fantastic8}.  

\subparagraph{1.\ The PDA for $E_{\mr{fff}}$.} 
We consider $(p,q,r)\in E_{\mr{fff}}$ and construct the PDA $P_{(p,q,r)}$. 
We have that  $I_{(p,q,r)}=\{\langle p,q,r\rangle \}$  contains an interaction between $s_i$ and $s_j$ if and only if 
\begin{align}
    \begin{aligned}
    i&=p+1 \\
    j&=p+q+2 \\ 
    n&=p+q+r+2        
    \end{aligned}
\end{align}
The PDA starts by reading the first $p+q+r+2$ symbols of its input and storing them in its state.
Then it checks if the next input symbol is $\bullet$.
If yes, then $n=p+q+r+2$, and it adds $h(s_i,s_j)$  to its stack.
The relevant spin values are stored in its states  and the  value $h(s_i,s_j)$ can be hardwired into the transition rules.
 If it reads $\bullet$ during any other step of the computation, it rejects. 

\subparagraph{2.\ The PDA for $E_{\mr{lff}}$.} 
We now consider   $((w_p,v_p), q, r) \in E_{\mr{lff}}$ and construct the PDA $P_{((w_p,v_p), q, r)}$. 
$s_i$ and $s_j$ interact if and only if 
\begin{align}\label{eqn:pff}
\begin{aligned}
    i &= v_p+1 \mod w_p \\
    j &= i+q+1 \\
    n &= j+r
\end{aligned}
\end{align}
If, for a given $n$, \cref{eqn:pff} has a solution, this solution is unique.
Hence, the PDA needs to compute at most one interaction, and works as follows.  
The PDA reads the first $r+q+2$ spin symbols $s_1\ldots s_{r+q+2}$ and stores them in its states. 
Then it iteratively reads the next input symbol, stores it in its states and removes the left-most of the currently stored input symbols---we call this the main iteration.
Note that at any given time, the stored spins are $s_i\ldots s_{i+r+q+1}$.
To test if $i = v_p+1 \mod w_p$, 
it uses a counter $c_{w_p}$ initialised at 1, which is updated as 
\begin{equation} 
c_{w_p}\mapsto c_{w_p}+1 \mod w_p
\end{equation}
at each step of the main iteration.  
Thus, if $i$ solves the first equation of \cref{eqn:pff}, $c_{w_p}= v_p+1$.
If this is the case, the PDA has stored $s_i\ldots s_n$, as by \cref{eqn:pff} $n=i+r+q+1$. 
It then checks if the next input symbol is $\bullet$.
If yes, it  adds $h(s_i,s_j)$ to its stack where, according to \cref{eqn:pff}, $j = n-r$. 
This is possible since at this step of the computation both $s_i$ and $s_j$ are stored in the states of the PDA.
If no, it continues with the main iteration. 
If it reaches $\bullet$ at any other step of the computation, it rejects.

\subparagraph{3.\ The PDA for $E_{\mr{flf}}$.} 
We now consider  $(p,(w_q,v_q),r))\in E_{\mr{flf}}$ and construct the PDA $P_{(p,(w_q,v_q),r))}$.
 $s_i$ and $s_j$ interact if and only if
\begin{align}\label{eqn:fpf}
\begin{aligned}
    i &= p+1 \\
    j &= i+v_q+1 \mod w_q  \\
    n &= j+r
\end{aligned}
\end{align}
The PDA starts by moving its head $p$ symbols to the right. Next it stores $s_{p+1}$ in its states, since according to \cref{eqn:fpf}, $i=p+1$.
It now enters the main iteration: 
It stores the next $r$ spin symbols in its states.
Its head is now placed over $s_{p+r+2}$, and it currently stores $s_i \ldots s_{i+r}$. 
At each step of the main iteration, it reads  the next spin symbol, stores it in its states and deletes the left-most spin symbol from its states.
In addition to that, it uses a counter $c_{w_q}$ that is initialised at
$c_{w_q}=i \mod w_q$. 
At each step of the main iteration the counter is updated as
\begin{equation}
c_{w_q}\mapsto c_{w_q}+1 \mod w_q
\end{equation} 
If during this main iteration the leftmost stored spin symbol is $s_j$ then the counter is $c_{w_q}=j \mod w_q$.
Once it reaches $\bullet$ the leftmost stored spin symbol $s_j$ satisfies $j=n-r$, i.e.\ solves the last equation in \cref{eqn:fpf}. If additionally $c_{w_q}=i+v_q+1$, it also solves the second equation in \cref{eqn:fpf}.
The PDA then adds $h(s_i,s_{n-r})$ to the stack. 
This is possible since both $s_i$ and $s_{n-r}$ are then stored in the states. 
If $c_{w_q}\neq j \mod w_q$, the input is rejected. 
If during any other step of the computation the PDA reaches $\bullet$, the input is rejected.

\subparagraph{4.\ The PDA for $E_{\mr{ffl}}$.} 
For $(p,q,(w_r,v_r)) \in E_{\mr{ffl}}$ we construct the PDA $P_{(p,q,(w_r,v_r))}$. 
$s_i$ and $s_j$ interact if and only if 
\begin{align}\label{eqn:ffp}
\begin{aligned}
    i &= p+1 \\
    j &= i+q+1  \\
    n &= j+v_r \mod w_r
\end{aligned}
\end{align}
Now both spin states $s_i,s_j$ can be read  from the beginning of the spin configuration. 
The PDA starts by storing the first $p+q+2$ spins from its input in its states.
Next it counts if $n=p+q+2+v_r \mod w_r$, again using $w_r$ of its states as a counter modulo $w_r$. 
If yes, it adds $h(s_i,s_j)$ to its stack; if no, it rejects. 

\subparagraph{5.\ The PDA for $E_{\mr{llf}}$.} 
We now consider $((w_p,v_p),(w_q,v_q),r) \in E_{\mr{llf}}$ and construct the PDA $P_{((w_p,v_p),(w_q,v_q),r)}$. 
$s_i$ and $s_j$ interact if and only if 
\begin{align}\label{eqn:ppf}
    \begin{aligned}
    i &= v_p+1  \mod w_p \\
    j &= i+v_q+1 \mod w_q \\
    n &= j+r 
    \end{aligned}
\end{align}
Let $l\coloneqq \mr{lcm}(w_p,w_q)$, $g\coloneqq \mr{gcd}(w_p,w_q)$. 
By the Chinese remainder theorem, if 
\begin{equation}
v_p+1= j-v_q-1 \mod g
 \end{equation}
 \cref{eqn:ppf} has a unique solution modulo $l$, else it has no solution. 
In particular, if there exists a solution, then there is a unique one with $i\leq l$. 
All further solutions $i^{\prime}$ are given as $i^{\prime}=i+m  l$ where
$m$ satisfies 
\begin{equation}\label{eq:iprime}
i+m  l\leq n-r-v_q-1
\end{equation}
The PDA first non-deterministically guesses both $s_{n-r}$ and the unique $i\leq l$ solving \cref{eqn:ppf}. Then it iterates over the input and adds all
$h(s_{i^{\prime}}, s_j)$ for $i^{\prime}=i+m  l$ satisfying \eqref{eq:iprime} to its stack.

More precisely,
it starts by non-deterministically guessing the state of spin $s_{n-r}$.
Next it reads the first $r+1$ symbols of the input and stores them in its states.  Now the main iteration starts.
At each step of the main iteration, the PDA reads the next symbol from the input, stores this symbol in its states and deletes the left-most of its stored symbols.
Additionally, it uses two counters, $c_{w_p}$ and $c_l$. 
Both counters are initialised as one. 
After each step of the iteration,
the counters are updated as 
\begin{align}
    \begin{aligned}
        c_{w_p} & \mapsto c_{w_p}+1 \mod w_p\\
        c_l & \mapsto c_l +1 \mod l
    \end{aligned}
\end{align}
Note that both these counters correspond to the position of the left-most stored spin symbol, i.e.\ when $s_i \ldots s_{i+r}$ are stored then $c_{w_p}=i \mod w_p$ and $c_l= i \mod l$.
The main iteration stops once $c_l=0$.

If, during the main iteration, $c_{w_p}=v_p+1$, the PDA non-deterministically branches into the two options of the current position $i$ either solving or not solving \cref{eqn:ppf}. 
If it  guesses that $i$ does not solve \cref{eqn:ppf}. 
It continues the main iteration with $i+1$.
If it guesses that $i$ solves \cref{eqn:ppf}
it adds $h(s_i,s_{n-r})$ to the stack. 
Next it sets $c_l$  to zero, and sets an additional counter $c_{w_q}$ modulo $w_q$ to zero, too. 
It continues iterating over the remaining input, still updating $c_l$ as before. 
Additionally, it now updates $c_{w_q} \mapsto c_{w_q}+1 \mod w_q$ instead of $c_{w_p}$.
Note that now both counters, $c_l$ and $c_{w_q}$ correspond to the position of the left-most stored spin symbol relative to $s_i$, i.e.\ when $s_{i^{\prime}} \ldots s_{i^{\prime}+r}$ is stored, the counters correspond to $c_l = i^{\prime}-i \mod l$, and similarly for $c_{w_q}$.
If, at any time $c_l = 0$, it adds the corresponding energy $h(s_{i^{\prime}},s_{n-r})$ to the stack, as in that case  $i^{\prime}= i +m  l$ and as $i$ solves \cref{eqn:ppf} so does $i^{\prime}$.
Given that the non-deterministic guess of $i$ solving \cref{eqn:ppf} was right,
the PDA hence accumulates the energy contributions of all solutions of \cref{eqn:ppf} on its stack. 

Finally, once the head reaches $\bullet$, it has $s_{n-r} \ldots s_n$ stored in its states and the second counter yields $c_{w_q}= j-i \mod w_q$.
This allows the PDA  to verify its two non-deterministic guesses. If $c_{w_q}=v_q+1$ and the initial guess of $s_{n-r}$ was correct, it accepts; else it rejects.

\subparagraph{6.\ The PDA for $E_{\mr{lfl}}$.} 
We now consider  $((w_p,v_p), q, ((w_r,v_r)) \in E_{\mr{lfl}}$ and construct the PDA $P_{((w_p,v_p), q, ((w_r,v_r))}$. 
 $s_i$ and $s_j$ interact if and only if 
\begin{align}\label{eqn:pfp}
    \begin{aligned}
        i &= v_p +1 \mod w_p \\
        j &= i + q + 1 \\
        n &= j + v_r \mod w_r
    \end{aligned}
\end{align}
Using again the Chinese remainder Theorem with $l\coloneqq \mr{lcm}(w_p,w_r)$, \cref{eqn:pfp} either has a unique solution modulo $l$, or there exists no solution. 
The PDA can now be built similarly to the previous case,
the only difference is that, as $j=i+q+1$, there is no need to apply non-determinism to obtain $s_j$. 

More precisely, the PDA first traverses the input string,
while keeping track of the current head position, using two modulo counters, $c_{w_p}$ and $c_l$, that are both initialised as $c_{w_p}=c_l=1$. This process stops when $c_l=0$.
Whenever $c_{w_p}=v_p+1$ it non-deterministically guesses if the current position $i$ solves \cref{eqn:pfp}. If no, it continues traversing the input; if yes, it    sets $c_l$ to zero,
reads and stores the next $q+1$ spin symbol $s_i \ldots s_{i+q+1}$. 
Then it adds $h(s_i,s_{i+q+1})$ to its stack and further initialises an additional modulo $w_r$ counter $c_{w_r}$ at zero.
Now it iterates over the remaining input. 
At each step it deletes the leftmost stored spin and  stores the next symbol from the input.
Additionally, the two counters are updated. 

Note that when the stored spins are $s_{i^{\prime}} \ldots s_{i^{\prime}+q+r}$, the values of the two counters are $c_l = i^{\prime}-i \mod l$, $c_{w_r} = i^{\prime}-i \mod w_r$.
While $c_l$ is used to obtain all further solutions $i^{\prime}= i+m   l$ to \cref{eqn:pfp}, $c_{w_r}$ allows the PDA to validate the non-deterministic guess of $i$ being a valid solution.
If at any time $c_l = 0$ the PDA adds $h(s_{i^{\prime}}, s_{i^{\prime}+q+1})$ to the stack. Finally, when it reaches $\bullet$, $c_{w_r}= n-i-q-1 \mod w_r$. 
Hence, $i$ solves \cref{eqn:pfp} if and only if $c_{w_r}=v_r$. 
If this holds, the PDA accepts; else the initial guess was wrong, and it rejects.

\subparagraph{7.\ The PDA for $E_{\mr{fll}}$.} 
We now consider $(p,(w_q,v_q),(w_r,v_q)) \in E_{\mr{fll}}$ and build the PDA $P_{(p,(w_q,v_q),(w_r,v_q))}$.  $s_i$ and $s_j$   interact if and only if 
\begin{align}\label{eqn:fpp}
    \begin{aligned}
    i &= p+1 \\
    j &= i+v_q+1 \mod w_q \\
    n &= j+v_r \mod w_r
    \end{aligned}
\end{align}
As before, let $l\coloneqq \mr{lcm}(w_q,w_r)$. The PDA starts by traversing the input until it reaches $s_{p+1}$. It then stores $s_{p+1}$ in its states.
Next it initialises two modulo counters $c_{w_q}$ and $c_l$ at $c_{w_q}=c_l=1$ and traverses the input until $c_l=0$. Whenever $c_{w_q}=v_q+1$, it non-deterministically guesses if the current
position $j$ solves \cref{eqn:fpp}. 
If no, it continues traversing the input as before; if yes it sets $c_l=0$, initialises a second modulo counter $c_{w_r}$ at zero and adds $h(s_{p+1},s_j)$ to its stack.
The PDA then traverses the remaining input.

Whenever for the current head position $j^{\prime}$, $c_l=0$ it adds  $h(s_{p+1},s_j^{\prime})$ to its stack. Finally, when its head reaches $\bullet$, 
it accepts if  $c_{w_q}=v_r$, and rejects otherwise. 
If at any other occasion its head reaches $\bullet$, it also rejects.

\subsection{Proof of \cref{thm:main} \ref{thm:main:ccs}}

First we prove that if ${\sf L}_{\mc{M}}$ is constructive context sensitive, then ${\sf E}_{\mc{M}}$ is context sensitive.
So let $\mc{M}$ be an Ising model with constructive context sensitive ${\sf L}_{\mc{M}}$, i.e.\ there exists a constructive multitape LBA $M$ that accepts ${\sf L}_{\mc{M}}$.
Similarly to \ref{thm:main:cf1:ii}, 
given $M$ we will build another LBA $M_C$ that on input 
$
0^{i-1}10^{j-1}10^{n-j}
$
computes $C_{\mc{M}}(n,i,j)$.

$M_C$ uses $4$ extra input tapes $T_{\mr{in},1}, \ldots ,T_{\mr{in},4}$ and $4$ energy tapes $T_{E,1}, \ldots ,T_{E,4}$. 
Using its input, $M_C$ writes the $4$ spin configurations
\begin{equation}
0^{i-1}10^{j-1}10^{n-j} , \quad
0^{i-1}10^{n-i}, \quad 
0^{j-1}10^{n-j}, \quad 
0^n
\end{equation} 
to the $4$ extra input tapes. This can be done by copying the input and replacing one or both $1$s with $0$s.
Now $M_C$ simulates $M$ on these $4$ input spin configurations, and the corresponding energies  obtained from $M$ are written to the $4$ energy tapes $T_{E,1}, \ldots ,T_{E,4}$. 
As $M$ is constructive,
it necessarily computes these $4$ energies before even considering any possible input energy $e$. 
Note that, as these $4$ energies are computed in binary, they can be stored within the LBA bounds, since for a spin configuration of length $n$,  the maximal absolute value of the energy is $\binom{n}{2}$. 
Finally, $M_C$ adds the $4$ stored energies; note that addition of binary numbers is possible for a LBA. 
This way, $M_C$ computes $C_{\mc{M}}( n, i,j)$, and it accepts the input if and only if $C_{\mc{M}}( n, i,j)=1$. 

Conversely, 
if ${\sf E}_{\mc{M}}$ is context sensitive, there exists a LBA $M_E$ that accepts ${\sf E}_{\mc{M}}$. 
According to the Immerman–Szelepcsényi theorem there exists another LBA $M_{\Bar{E}}$
that accepts the complement of ${\sf E}_{\mc{M}}$. 
We now build a constructive LBA $M$ that accepts ${\sf L}_{\mc{M}}$ as follows.

$M$ uses a specific tape $T_{S_E}$ to simulate $M_E$ and a second  tape $T_{S_{\Bar{E}}}$ to simulate $M_{\Bar{E}}$. In addition, $M$ has the usual input tape $T_{\mr{in}}$ and energy tape $T_E$. 
On input $s_1 \ldots s_n\bullet e$, 
$M$ iterates over all possible pairs $(i,j)$ that satisfy $i< j \leq n$. 
This can be achieved by marking the input on $T_{\mr{in}}$ appropriately; explicitly, by marking it as
\begin{equation}
s_1^{\prime}s_2^{\prime \prime}s_3 \ldots s_n\bullet e
\end{equation} 
in the beginning, then moving the $s_j^{\prime \prime}$ mark one step to the right after each iteration. 
Once $\bullet$ is reached the $s_j^{\prime \prime}$ mark is removed, $M$ moves the $s_i^{\prime}$ mark one position to the right and marks the spin symbol to the right of it as $s_{i+1}^{\prime \prime} \bullet e$.

At every step of the iteration, with $s_i$ and $s_j$ marked, $M$ copies the entire input spin configuration \begin{equation}
s_1 \ldots s_i^{\prime} \ldots s_j^{\prime \prime} \ldots s_n
\end{equation} 
both to $T_{S_E}$ and $T_{S_{\Bar{E}}}$,
and replaces each unmarked spin with a $0$ and each marked spin with a $1$.
Thereby the edge between $i$ and $j$ is written to these two tapes.

Now $M$ simulates $M_E$ with $T_{S_E}$ as input and $M_{\Bar{E}}$ with $T_{S_{\Bar{E}}}$ as input. 
If $M_E$ accepts, $\mc{M}$ contains an interaction of $s_i$ and $s_j$, 
so $M$ adds $h(s_i,s_j)$ in binary to its energy tape.
If $M_{\Bar{E}}$ accepts, $\mc{M}$ contains no interaction of $s_i$ and $s_j$,  
so $M$ moves to the next pair of spins without adding $h(s_i,s_j)$ to $T_E$.
After that, $M$ clears both $T_{S_E}$ and $T_{S_{\Bar{E}}}$, before it continues with the next step of the iteration.

When the iteration terminates, $M$ has stored $ H_{\mc{M}}(s_1 \ldots s_n)$ in binary on $T_E$ and compares this to the input energy $e$.
If the computed and the input energy coincide, it accepts; otherwise, it rejects.

If  $n \notin N_{\mc{M}}$, the iteration terminates without $M_E$ accepting a single edge, in which case $M$ rejects the input.  
This case can be checked by, prior to the iteration over possible edges, writing a distinguished symbol to the energy tape that is removed once the first energy is added.  

\subsection{Proof of \cref{thm:main} \ref{thm:main:undec}}

The proof follows a similar line of reasoning as that  of \ref{thm:main:ccs}.
If ${\sf L}_{\mc{M}}$ is decidable there exists a Turing machine $M$ that decides it. 
From this we can build a second Turing machine $M_H$ that computes $H_{\mc{M}}$ as a function, i.e.\ that on input
$s_1 \ldots s_n$ with $n \in N_{\mc{M}}$, after a finite number of steps, halts with $H_{\mc{M}}(s_1 \ldots s_n)$ written to its output  tape $T_{\mr{out}}$. This can be achieved as follows:
On input $s_1 \ldots s_n$, $M_H$ iterates over 
\begin{equation}
e \in u\left(\left\{-\binom{n}{2}, \ldots ,\binom{n}{2}\right\}\right)
\end{equation} 
i.e.\ over all possible energies that the input may have. 
For each, $M_H$ writes $s_1 \ldots s_n\bullet e$ to an additional tape $T_S$ and simulates $T$ with $T_S$ as input. As $M$ decides ${\sf L}_{\mc{M}}$,
this simulation halts after a finite number of steps. 
If $M$ accepts $s_1 \ldots s_n\bullet e$, then $e$ is the correct energy, i.e.\ 
\begin{equation}
e = H_{\mc{M}}(s_1 \ldots s_n)
\end{equation} 
and $M_H$ writes $e$ to $T_{\mr{out}}$ and halts.
If $M$ rejects, $M_H$ continues with the next step of the iteration. 
If the iteration terminates without $M_H$ halting, then $n\notin N_{\mc{M}}$ and $M_H$ halts and rejects.

Using $M_H$, we can proceed analogously to the proof of \ref{thm:main:ccs},  in order to build a Turing machine $M_C$ that on input $0^{i-1}10^{j-1}10^{n-j}$ computes $C_{\mc{M}}(n,i,j)$.
So $M_C$ decides ${\sf E}_{\mc{M}}$.

Conversely, if ${\sf E}_{\mc{M}}$ is decidable, as in the proof of \ref{thm:main:ccs} we can build a Turing machine $M$ that on input $s_1 \ldots s_n\bullet e$ iterates over all possible edges $(i,j)$ with $i < j \leq n$ and uses the decider for ${\sf E}_{\mc{M}}$, $M_E$, to decide if $(i,j) \in (E_{\mc{M}})_n$. 
If yes,  $M$ adds $h(s_i,s_j)$ to its energy tape $T_E$ and continues the iteration. 
If no, $M$   continues the iteration without adding the energy.
Once the iteration over edges terminates, $M$ has $H_{\mc{M}}(s_1 \ldots s_n)$ stored on its energy tape. Hence, ${\sf L}_{\mc{M}}$ can be decided by letting $M$ accept the input if and only if there was at least
one edge accepted by $M_E$, to ensure that $n \in N_{\mc{M}}$, and additionally input energy and computed energy are equal.

\section{Formal language theory toolbox}
\label{ssec:FormalLanguage}

Let $\Sigma$ be a finite set, called the alphabet.
Let $\Sigma^*$ denote the free monoid over $\Sigma$, with unit being the empty string $\epsilon$.
In the context of formal language theory $\Sigma^*$ is often called the Kleene star of $\Sigma$. $\Sigma^*$ contains all finite strings, that can be formed with symbols from 
$\Sigma$, including the empty string.
A formal language ${\sf L}$  over $\Sigma$ is a subset of $\Sigma^*$.

While most languages can only be characterised in an extensive way, 
namely by specifying the (infinite) set ${\sf L}\subseteq \Sigma^*$ itself, 
some admit a finite description.
There are two ways of providing this finite description: 
by providing a grammar $G$ that generates ${\sf L}$,
 or by constructing an automaton that accepts ${\sf L}$.

\begin{definition}[Grammar]\label{def:grammar}
A \emph{grammar} is a $4$-tuple $G=(S,T,NT,P)$, where 
\begin{itemize}
    \item $S \in NT$ is a distinguished symbol, called the start symbol of $G$;
    \item $T$, $NT$ are disjoint, finite sets, whose elements are called terminal and non-terminal symbols respectively;
    \item $P\subseteq (T \cup NT)^* \times (T \cup NT)^*$ is a finite set of production rules. For $(\alpha, \beta) \in P$ we write $\alpha \rightarrow \beta $ and for $\{(\alpha, \beta), (\alpha, \beta^{\prime}), (\alpha, \beta^{\prime \prime})\} \subseteq P$ we write $\alpha \rightarrow \beta \mid \beta^{\prime} \mid \beta^{\prime \prime}$.
\end{itemize}
Given a string $w \in (T \cup NT)^*$, a production rule $\alpha \rightarrow \beta \in P$ is applied to it by replacing an occurrence of $\alpha$ as a substring of $w$ with $\beta$.
If a string $w$ can be obtained from another string $w^{\prime}$ by repeated application of production rules of $G$ we say that $w$ can be derived from $w^{\prime}$ by means of $G$, and write $w^{\prime} \Rightarrow^*_G w$.
The language ${\sf L}(G)$ that a grammar $G$ generates is the set of all terminal strings that can be derived from the start symbol $S$, 
\begin{align}
    {\sf L}(G) \coloneqq \{w \in T^* \mid S \Rightarrow_G^* w \} 
\end{align}
\end{definition}
Grammars can be classified according to the form of the production rules they contain. The most famous such classification is the Chomsky hierarchy.
\begin{definition}[Chomsky hierarchy]\label{def:Chomsky}
The \emph{Chomsky hierarchy} is the inclusion hierarchy of formal grammars consisting of the following four types of formal grammars;
\begin{align}
    \text{regular} \subset \text{context free} \subset \text{context sensitive} \subset \text{unrestricted}
\end{align}
where all inclusions are strict \cite{Ch63, Ch65}. 
A grammar $G$ belongs to either of these types if for all production rules $\alpha\to\beta \in P$, the following holds: 
\begin{description}[before={\renewcommand\makelabel[1]{\bfseries ##1}}]
    \item[regular] if $\alpha \in NT$ and 
$\beta = \epsilon$ or 
$\beta \in T$  or 
$\beta = b B$ with $b \in T$ and $B \in NT$; 
    \item[context free] if $\alpha \in NT$; 
    \item[context sensitive] if $aBc \rightarrow adc$, where $a,c \in (T \cup NT)^*$, $B \in NT$, and $\epsilon \neq d \in (T \cup NT)^*$; 
    \item[unrestricted] in any  case.
\end{description}
This hierarchy of grammars can be lifted to a hierarchy of languages, by calling a formal language ${\sf L}$ regular if there exists a regular grammar with ${\sf L} = {\sf L}(G)$, and similar for context free and context sensitive. 
The class of languages corresponding to unrestricted grammars is called recursively enumerable. 
If both a language and its complement are recursively enumerable, then it is called decidable. 
\end{definition}

For every level in the Chomsky hierarchy, there exists a type of automaton (i.e.\ a model of computation) that accepts the languages from that level: 
regular languages are accepted by finite state automata (FSA), 
context free languages are accepted by pushdown automata (PDA),
 context sensitive languages are accepted by linear bounded automata (LBA), 
and recursively enumerable language are accepted by Turing machines (TM).
Proving that a language is accepted by a certain type automaton is equivalent to proving that it is in the corresponding level. 
We now review these automata (see e.g.\ \cite{Ko97,Ro971}).

A FSA can be imagined as a machine with one  tape, and a head that scans one cell of the input tape at a time. 
The FSA has a finite number of states in its head as memory. 
The computation starts with the input written on the  tape and the head placed over the first input symbol.
At each computation step, it reads the symbol that its head is currently placed over, 
and, depending on the symbol on the tape and the current state, transitions to a new state. 
The head then moves to the next input symbol. 
A FSA can neither change the direction of its head movement nor overwrite the tape. 
\begin{definition}(Finite state automaton)\label{def:FSA}
A \emph{finite state automaton} is a $5$-tuple $F = (Q, \Sigma, \delta, q_0, A)$, where 
\begin{itemize}
    \item $Q$ and $\Sigma$ are finite sets called the states and the input alphabet; 
    \item $\delta: Q \times \Sigma \rightarrow Q$ is called the transition function; 
    \item $q_0 \in Q$ is the start state; 
    \item $A\subseteq Q$ are the accept states. 
\end{itemize}
\end{definition}
The transition function encodes one computation step of the FSA $F$: 
When in state $q$ upon reading $s$, $F$ transitions to state $q^{\prime} = \delta(q,s)$.
On input $w_1\ldots w_n \in \Sigma^*$, $F$ starts in state $q_0$. 
It then processes the entire input: for each input symbol it uses $\delta$ and the current state to compute the new state; then it moves on with the next input symbol.
After processing the entire input, 
if $F$ is in a state $f\in A$, the input is accepted by $F$; 
else, $F$ rejects.

A PDA can be imagined as a FSA which additionally has access to a stack. 
Specifically, at each step of the computation, 
the head of the PDA reads the current symbol on the tape, 
pops a symbol from the top of the stack, 
and pushes a finite number of symbols onto the stack. 
Then the head moves to the next symbol.
\begin{definition}[Pushdown automaton]\label{def:PDA}
A \emph{pushdown automaton} is a $7$-tuple $P=(Q, \Sigma_{\rm in}, \Sigma_{\rm stack}, \delta, q_0, Z, A)$, where
\begin{itemize}
    \item $Q,\Sigma_{\rm in}$ and $\Sigma_{\rm stack}$ are finite sets called the states, input alphabet and stack alphabet; 
    \item $\delta \subseteq (Q \times (\Sigma_{\rm in} \cup \{\epsilon\}) \times \Sigma_{\rm stack} ) \times (Q \times \Sigma_{\rm stack}^*)$ is a finite set called the transition relation;  
    \item $q_0 \in Q$ is the initial state; 
    \item $Z \in \Sigma_{\rm stack}$ is the initial stack symbol; 
    \item $A \subseteq Q$ are the accept states. 
\end{itemize}
\end{definition}
The transition relation models one step of the computation:
When in state $q$, upon reading $x$ and popping $s$, $P$ transitions to state $q^{\prime}$ and pushes $s^{\prime}$ to the stack, where $q^{\prime}$ and $s^{\prime}$ are such that 
\begin{equation}
(q,x,s, q^{\prime}, s^{\prime}) \in \delta
\end{equation}
Then $P$ moves its head to the next input symbol.
If there are multiple such $q^{\prime}$, $s^{\prime}$, 
$P$ branches its computation to pursue all such options simultaneously. 

$P$ is called deterministic if for any $q \in Q$, $x \in \Sigma_{\rm in}$ and $s \in \Sigma_{\rm stack}$, there exist a unique $q^{\prime} \in Q$, $s^{\prime} \in \Sigma_{\rm stack}^*$, such that either 
\begin{equation}
(q,x,s,q^{\prime}, s^{\prime})\in \delta \quad \text{or} \quad (q,\epsilon,s,q^{\prime}, s^{\prime})\in \delta
\end{equation}
Otherwise $P$ is called non-deterministic.

An input string $w_1\ldots w_n \in \Sigma_{\rm in}^*$ is accepted by $P$ if 
with its head placed over the first symbol $w_1$, in state $q_0$ and with its stack containing only one symbol $Z$, after processing the entire input as dictated by $\delta$,
there is at least one computation path that leads to a state $f \in A$.
Otherwise, the input is rejected by $P$.

The most powerful notion of machine that we will encounter in this work is that of a Turing machine (TM). 
A TM can be imagined as a machine with a finite number of states and 
 an input  tape. In contrast to a PDA, a TM can overwrite the input tape, and move left or right. 
\begin{definition}[Turing machine]\label{def:TM}
A \emph{Turing machine} (TM) is a $7$-tuple $M = (Q, \Sigma_{\rm in}, \Sigma_{\rm tape}, \delta, q_0, A, B)$, where 
\begin{itemize}
    \item $Q,\Sigma_{\rm tape}$ are finite sets called the states and the tape alphabet; 
    \item $\Sigma_{\rm in} \subseteq \Sigma_{\rm tape}$ is  the input alphabet; 
    \item $\delta \subseteq (Q \times \Sigma_{\rm tape}) \times (Q\times \Sigma_{\rm tape} \times \{L,R\})$ is the  transition relation; 
    \item $q_0 \in Q$ is the start state; 
    \item $A \subseteq Q$ are the final states; 
    \item $B \in \Sigma_{\rm tape}$ is the blank symbol that represents an empty input cell. 
\end{itemize}
\end{definition}
When in state $q$ upon reading $s$, $M$ transitions to state $q^{\prime}$, overwrites $s$ with $s^{\prime}$ and moves its head one step in the direction specified by $m \in \{L,R\}$, where $(q^{\prime}, s^{\prime}, m)$ are specified by $\delta$, 
\begin{equation} \label{eq:transitionTM}
(q,s,q^{\prime}, s^{\prime}, m)\in \delta
\end{equation} 
If there are multiple such options $M$ branches its computation path to carry them out  simultaneously.

$M$ is called deterministic if for each $(q,s) \in Q \times  \Sigma_{\rm tape}$ there exists at most one $(q^{\prime}, s^{\prime}, m) \in Q \times \Sigma_{\rm tape} \times \{L,R\}$
such that \eqref{eq:transitionTM} holds. 
Otherwise, $M$ is called non-deterministic.

If for a given state and input symbol $(q,s)$  there is no  $(q^{\prime}, s^{\prime}, m)$ satisfying \eqref{eq:transitionTM}, $M$ is said to halt in state $q$.
An input string $w_1\ldots w_n \in \Sigma_{\rm in}$ is  accepted by $M$ if $M$ when started in state $q_0$ with $w_1\ldots w_n$ written on its input tape and its head placed over the first cell, after repeatedly performing the transitions as specified by $\delta$,
after a finite number of steps there is at least one computation path that leads to $M$ halting in a final state.

Whereas a TM may use an unbounded amount of tape to carry out the computation, for the weaker notion of a linear bounded automaton (LBA) the accessible tape  is limited to the cells which are initially used by the input string.
\begin{definition}[Linear bounded automaton]\label{def:LBA}
A \emph{linear bounded automaton} is a $9$-tuple $L = (Q, \Sigma_{\rm in}, \Sigma_{\rm tape}, \delta, q_0, A, B, \bot_L, \bot_R)$, where 
\begin{itemize}
    \item $(Q, \Sigma_{\rm in}, \Sigma_{\rm tape}, \delta, q_0, A, B)$ is a Turing machine
    \item $\bot_L, \bot_R \in \Sigma_{\rm tape}$ are two special symbols that satisfy
    \begin{align}
        \begin{aligned}
            (q,\bot_L,q^{\prime}, s^{\prime}, m)\in \delta & \Rightarrow s^{\prime} = \bot_L \ \text{and} \ m = R \\
             (q,\bot_R,q^{\prime}, s^{\prime}, m)\in \delta & \Rightarrow s^{\prime} = \bot_R \ \text{and} \ m = L 
        \end{aligned}
    \end{align}
\end{itemize}
\end{definition}
The special symbols $\bot_L, \bot_R$ serve as left and right endmarkers of the tape. 
Throughout the  computation, $L$ neither overwrites these endmarkers nor moves its head past them. 
Other than  that, the computation works exactly like that of a TM. 

Relaxing the definition of the LBA such that the accessible tape space is a linear function of the input length, or allowing the LBA to perform its computation on multiple tapes does not change the class of problems it can solve  \cite[Theorem 12]{Ho79}. 
Hence, the class of  context sensitive languages is identical with the complexity class \textsc{NLINSPACE} of  problems that can be solved in linear space on a non-deterministic Turing machine  \cite[Theorem 3.33]{Ro05}. 

Finally, we stress that most languages do not have a grammar (or, equivalently, are not recognised by a Turing machine), as there are uncountably many languages but countably many grammars (or Turing machines). Explicitly, the number of languages over a finite alphabet $\Sigma$  is $|\wp(\Sigma^*)| = 2^{|\mathbb{N}|}$, whereas the number of grammars (or Turing machines) is $\vert \mathbb{N} \vert$.


\begin{thebibliography}{43}%
\makeatletter
\providecommand \@ifxundefined [1]{%
 \@ifx{#1\undefined}
}%
\providecommand \@ifnum [1]{%
 \ifnum #1\expandafter \@firstoftwo
 \else \expandafter \@secondoftwo
 \fi
}%
\providecommand \@ifx [1]{%
 \ifx #1\expandafter \@firstoftwo
 \else \expandafter \@secondoftwo
 \fi
}%
\providecommand \natexlab [1]{#1}%
\providecommand \enquote  [1]{``#1''}%
\providecommand \bibnamefont  [1]{#1}%
\providecommand \bibfnamefont [1]{#1}%
\providecommand \citenamefont [1]{#1}%
\providecommand \href@noop [0]{\@secondoftwo}%
\providecommand \href [0]{\begingroup \@sanitize@url \@href}%
\providecommand \@href[1]{\@@startlink{#1}\@@href}%
\providecommand \@@href[1]{\endgroup#1\@@endlink}%
\providecommand \@sanitize@url [0]{\catcode `\\12\catcode `\$12\catcode
  `\&12\catcode `\#12\catcode `\^12\catcode `\_12\catcode `\%12\relax}%
\providecommand \@@startlink[1]{}%
\providecommand \@@endlink[0]{}%
\providecommand \url  [0]{\begingroup\@sanitize@url \@url }%
\providecommand \@url [1]{\endgroup\@href {#1}{\urlprefix }}%
\providecommand \urlprefix  [0]{URL }%
\providecommand \Eprint [0]{\href }%
\providecommand \doibase [0]{https://doi.org/}%
\providecommand \selectlanguage [0]{\@gobble}%
\providecommand \bibinfo  [0]{\@secondoftwo}%
\providecommand \bibfield  [0]{\@secondoftwo}%
\providecommand \translation [1]{[#1]}%
\providecommand \BibitemOpen [0]{}%
\providecommand \bibitemStop [0]{}%
\providecommand \bibitemNoStop [0]{.\EOS\space}%
\providecommand \EOS [0]{\spacefactor3000\relax}%
\providecommand \BibitemShut  [1]{\csname bibitem#1\endcsname}%
\let\auto@bib@innerbib\@empty
\bibitem [{\citenamefont {Ising}(1925)}]{Is25}%
  \BibitemOpen
  \bibfield  {author} {\bibinfo {author} {\bibfnamefont {E.}~\bibnamefont
  {Ising}},\ }\bibfield  {title} {\bibinfo {title} {{Beitrag zur Theorie des
  Ferromagnetismus}},\ }\href {https://doi.org/10.1007/BF02980577} {\bibfield
  {journal} {\bibinfo  {journal} {Zeitschrift f{\"{u}}r Physik}\ }\textbf
  {\bibinfo {volume} {31}},\ \bibinfo {pages} {253-258} (\bibinfo {year}
  {1925})}\BibitemShut {NoStop}%
\bibitem [{\citenamefont {Pathria}(1996)}]{Pa96}%
  \BibitemOpen
  \bibfield  {author} {\bibinfo {author} {\bibfnamefont {R.~K.}\ \bibnamefont
  {Pathria}},\ }\href@noop {} {\emph {\bibinfo {title} {{Statistical
  Mechanics}}}}\ (\bibinfo  {publisher} {2nd edn. Amsterdam, the Netherlands: Elsevier},\ \bibinfo
  {year} {1996})\BibitemShut {NoStop}%
\bibitem [{\citenamefont {McCoy}\ and\ \citenamefont {Wu}(2013)}]{Mc13}%
  \BibitemOpen
  \bibfield  {author} {\bibinfo {author} {\bibfnamefont {B.}~\bibnamefont
  {McCoy}}\ and\ \bibinfo {author} {\bibfnamefont {T.~T.}\ \bibnamefont {Wu}},\
  }\href@noop {} {\emph {\bibinfo {title} {{The Two-Dimensional Ising
  model}}}}\ (\bibinfo  {publisher} {2nd edn. Mineola, NY: Dover},\
  \bibinfo {year} {2013})\BibitemShut {NoStop}%
\bibitem [{\citenamefont {Ambj{\o}rn}\ \emph {et~al.}(2009)\citenamefont
  {Ambj{\o}rn}, \citenamefont {Anagnostopoulos}, \citenamefont {Loll},\ and\
  \citenamefont {Pushinka}}]{Am09a}%
  \BibitemOpen
  \bibfield  {author} {\bibinfo {author} {\bibfnamefont {J.}~\bibnamefont
  {Ambj{\o}rn}}, \bibinfo {author} {\bibfnamefont {K.~N.}\ \bibnamefont
  {Anagnostopoulos}}, \bibinfo {author} {\bibfnamefont {R.}~\bibnamefont
  {Loll}},\ and\ \bibinfo {author} {\bibfnamefont {I.}~\bibnamefont
  {Pushinka}},\ }\bibfield  {title} {\bibinfo {title} {{Shaken, but not
  stirred---Potts model coupled to quantum gravity}},\ }\href
  {https://doi.org/10.1016/j.nuclphysb.2008.08.030} {\bibfield  {journal}
  {\bibinfo  {journal} {Nucl. Phys. B}\ }\textbf {\bibinfo {volume} {807}},\
  \bibinfo {pages} {251-264} (\bibinfo {year} {2009})}\BibitemShut {NoStop}%
\bibitem [{\citenamefont {Chandler}(1987)}]{Ch87}%
  \BibitemOpen
  \bibfield  {author} {\bibinfo {author} {\bibfnamefont {D.}~\bibnamefont
  {Chandler}},\ }\href@noop {} {\emph {\bibinfo {title} {{Introduction to
  Modern Statistical Mechanics}}}}\ (\bibinfo  {publisher} {New York, NY: Oxford University
  Press},\ \bibinfo {year} {1987})\BibitemShut {NoStop}%
\bibitem [{\citenamefont {Kauffman}(2001)}]{Ka01}%
  \BibitemOpen
  \bibfield  {author} {\bibinfo {author} {\bibfnamefont {L.~H.}\ \bibnamefont
  {Kauffman}},\ }\href@noop {} {\emph {\bibinfo {title} {{Knots and
  Physics}}}}\ (\bibinfo  {publisher} {Singapore: World Scientific},\ \bibinfo
  {year} {2001})\BibitemShut {NoStop}%
\bibitem [{\citenamefont {Hopfield}(1982)}]{Ho82}%
  \BibitemOpen
  \bibfield  {author} {\bibinfo {author} {\bibfnamefont {J.~J.}\ \bibnamefont
  {Hopfield}},\ }\bibfield  {title} {\bibinfo {title} {{Neural networks and
  physical systems with emergent collective computational abilities}},\ }\href
  {https://doi.org/10.1073/pnas.79.8.2554} {\bibfield  {journal} {\bibinfo
  {journal} {Proc. Natl. Acad. Sci.}\ }\textbf {\bibinfo {volume} {79}},\
  \bibinfo {pages} {2554-2558} (\bibinfo {year} {1982})}\BibitemShut {NoStop}%
\bibitem [{\citenamefont {Amit}\ \emph {et~al.}(1985)\citenamefont {Amit},
  \citenamefont {Gutfreund},\ and\ \citenamefont {Sompolinsky}}]{Am85}%
  \BibitemOpen
  \bibfield  {author} {\bibinfo {author} {\bibfnamefont {D.~J.}\ \bibnamefont
  {Amit}}, \bibinfo {author} {\bibfnamefont {H.}~\bibnamefont {Gutfreund}},\
  and\ \bibinfo {author} {\bibfnamefont {H.}~\bibnamefont {Sompolinsky}},\
  }\bibfield  {title} {\bibinfo {title} {{Spin-glass models of neural
  networks}},\ }\href {https://doi.org/10.1103/PhysRevA.32.1007} {\bibfield
  {journal} {\bibinfo  {journal} {Phys. Rev. A}\ }\textbf {\bibinfo {volume}
  {32}},\ \bibinfo {pages} {1007-1018} (\bibinfo {year} {1985})}\BibitemShut
  {NoStop}%
\bibitem [{\citenamefont {Sol{\'{e}}}\ and\ \citenamefont
  {Goodwin}(2000)}]{So00}%
  \BibitemOpen
  \bibfield  {author} {\bibinfo {author} {\bibfnamefont {R.~V.}\ \bibnamefont
  {Sol{\'{e}}}}\ and\ \bibinfo {author} {\bibfnamefont {B.}~\bibnamefont
  {Goodwin}},\ }\href@noop {} {\emph {\bibinfo {title} {{Signs of Life: How
  complexity pervades biology}}}}\ (\bibinfo  {publisher} {New York, NY: Basic Books},\
  \bibinfo {year} {2000})\BibitemShut {NoStop}%
\bibitem [{\citenamefont {Bialeka}\ \emph {et~al.}(2012)\citenamefont
  {Bialeka}, \citenamefont {Cavagna}, \citenamefont {Giardina}, \citenamefont
  {Mora}, \citenamefont {Silvestri}, \citenamefont {Viale},\ and\ \citenamefont
  {Walczak}}]{Bi12}%
  \BibitemOpen
  \bibfield  {author} {\bibinfo {author} {\bibfnamefont {W.}~\bibnamefont
  {Bialeka}}, \bibinfo {author} {\bibfnamefont {A.}~\bibnamefont {Cavagna}},
  \bibinfo {author} {\bibfnamefont {I.}~\bibnamefont {Giardina}}, \bibinfo
  {author} {\bibfnamefont {T.}~\bibnamefont {Mora}}, \bibinfo {author}
  {\bibfnamefont {E.}~\bibnamefont {Silvestri}}, \bibinfo {author}
  {\bibfnamefont {M.}~\bibnamefont {Viale}},\ and\ \bibinfo {author}
  {\bibfnamefont {A.~M.}\ \bibnamefont {Walczak}},\ }\bibfield  {title}
  {\bibinfo {title} {{Statistical mechanics for natural flocks of birds}},\
  }\href {https://doi.org/10.1073/pnas.1118633109} {\bibfield  {journal}
  {\bibinfo  {journal} {Proc. Natl. Acad. Sci.}\ }\textbf {\bibinfo {volume}
  {109}},\ \bibinfo {pages} {4786-4791} (\bibinfo {year} {2012})}\BibitemShut
  {NoStop}%
\bibitem [{\citenamefont {Anderson}(1983)}]{An83}%
  \BibitemOpen
  \bibfield  {author} {\bibinfo {author} {\bibfnamefont {P.~W.}\ \bibnamefont
  {Anderson}},\ }\bibfield  {title} {\bibinfo {title} {{Suggested model for
  prebiotic evolution: The use of chaos}},\ }\href
  {https://doi.org/10.1073/pnas.80.11.3386} {\bibfield  {journal} {\bibinfo
  {journal} {Proc. Natl. Acad. Sci. USA}\ }\textbf {\bibinfo {volume} {80}},\
  \bibinfo {pages} {3386-3390} (\bibinfo {year} {1983})}\BibitemShut {NoStop}%
\bibitem [{\citenamefont {Tarazona}(1992)}]{Ta92}%
  \BibitemOpen
  \bibfield  {author} {\bibinfo {author} {\bibfnamefont {P.}~\bibnamefont
  {Tarazona}},\ }\bibfield  {title} {\bibinfo {title} {{Error thresholds for
  molecular quasispecies as phase transitions: From simple landscapes to
  spin-glass models}},\ }\href {https://doi.org/10.1103/PhysRevA.45.6038}
  {\bibfield  {journal} {\bibinfo  {journal} {Phys. Rev. A}\ }\textbf {\bibinfo
  {volume} {45}},\ \bibinfo {pages} {6038-6050} (\bibinfo {year}
  {1992})}\BibitemShut {NoStop}%
\bibitem [{\citenamefont {Lezon}\ \emph {et~al.}(2006)\citenamefont {Lezon},
  \citenamefont {Banavar}, \citenamefont {Cieplak}, \citenamefont {Maritan},\
  and\ \citenamefont {Fedoroff}}]{Le06}%
  \BibitemOpen
  \bibfield  {author} {\bibinfo {author} {\bibfnamefont {T.~R.}\ \bibnamefont
  {Lezon}}, \bibinfo {author} {\bibfnamefont {J.~R.}\ \bibnamefont {Banavar}},
  \bibinfo {author} {\bibfnamefont {M.}~\bibnamefont {Cieplak}}, \bibinfo
  {author} {\bibfnamefont {A.}~\bibnamefont {Maritan}},\ and\ \bibinfo {author}
  {\bibfnamefont {N.~V.}\ \bibnamefont {Fedoroff}},\ }\bibfield  {title}
  {\bibinfo {title} {{Using the principle of entropy maximization to infer
  genetic interaction networks from gene expression patterns}},\ }\href
  {https://doi.org/10.1073/pnas.0609152103} {\bibfield  {journal} {\bibinfo
  {journal} {Proc. Natl. Acad. Sci.}\ }\textbf {\bibinfo {volume} {103}},\
  \bibinfo {pages} {19033-19038} (\bibinfo {year} {2006})}\BibitemShut {NoStop}%
\bibitem [{\citenamefont {Bakk}\ and\ \citenamefont {H{\o}ye}(2003)}]{Ba03}%
  \BibitemOpen
  \bibfield  {author} {\bibinfo {author} {\bibfnamefont {A.}~\bibnamefont
  {Bakk}}\ and\ \bibinfo {author} {\bibfnamefont {J.~S.}\ \bibnamefont
  {H{\o}ye}},\ }\bibfield  {title} {\bibinfo {title} {{One-dimensional Ising
  model applied to protein folding}},\ }\href
  {https://doi.org/10.1016/S0378-4371(03)00018-9} {\bibfield  {journal}
  {\bibinfo  {journal} {Physica A}\ }\textbf {\bibinfo {volume} {323}},\
  \bibinfo {pages} {504-518} (\bibinfo {year} {2003})}\BibitemShut {NoStop}%
\bibitem [{\citenamefont {Ekeberg}\ \emph {et~al.}(2013)\citenamefont
  {Ekeberg}, \citenamefont {L{\"{o}}vkvist}, \citenamefont {Lan}, \citenamefont
  {Weigt},\ and\ \citenamefont {Aurell}}]{Ek13}%
  \BibitemOpen
  \bibfield  {author} {\bibinfo {author} {\bibfnamefont {M.}~\bibnamefont
  {Ekeberg}}, \bibinfo {author} {\bibfnamefont {C.}~\bibnamefont
  {L{\"{o}}vkvist}}, \bibinfo {author} {\bibfnamefont {Y.}~\bibnamefont {Lan}},
  \bibinfo {author} {\bibfnamefont {M.}~\bibnamefont {Weigt}},\ and\ \bibinfo
  {author} {\bibfnamefont {E.}~\bibnamefont {Aurell}},\ }\bibfield  {title}
  {\bibinfo {title} {{Improved contact prediction in proteins: Using
  pseudolikelihoods to infer Potts models}},\ }\href
  {https://doi.org/10.1103/PhysRevE.87.012707} {\bibfield  {journal} {\bibinfo
  {journal} {Phys. Rev. E}\ }\textbf {\bibinfo {volume} {87}},\ \bibinfo
  {pages} {012707} (\bibinfo {year} {2013})}\BibitemShut {NoStop}%
\bibitem [{\citenamefont {Leuth{\"{a}}usser}(1986)}]{Le86}%
  \BibitemOpen
  \bibfield  {author} {\bibinfo {author} {\bibfnamefont {I.}~\bibnamefont
  {Leuth{\"{a}}usser}},\ }\bibfield  {title} {\bibinfo {title} {{An exact
  correspondence between Eigen's evolution model and a two‐dimensional Ising
  system}},\ }\href {https://doi.org/10.1063/1.450436} {\bibfield  {journal}
  {\bibinfo  {journal} {J. Chem. Phys.}\ }\textbf {\bibinfo {volume} {84}},\
  \bibinfo {pages} {1884-1885} (\bibinfo {year} {1986})}\BibitemShut {NoStop}%
\bibitem [{\citenamefont {Leuth{\"{a}}usser}(1987)}]{Le87}%
  \BibitemOpen
  \bibfield  {author} {\bibinfo {author} {\bibfnamefont {I.}~\bibnamefont
  {Leuth{\"{a}}usser}},\ }\bibfield  {title} {\bibinfo {title} {{Statistical
  mechanics of Eigen's evolution model}},\ }\href
  {https://doi.org/10.1007/BF01010413} {\bibfield  {journal} {\bibinfo
  {journal} {J. Stat. Mech.}\ }\textbf {\bibinfo {volume} {48}},\ \bibinfo
  {pages} {343-360} (\bibinfo {year} {1987})}\BibitemShut {NoStop}%
\bibitem [{\citenamefont {Rizzato}\ \emph {et~al.}(2020)\citenamefont
  {Rizzato}, \citenamefont {Coucke}, \citenamefont {de~Leonardis},
  \citenamefont {Barton}, \citenamefont {Tubiana}, \citenamefont {Monasson},\
  and\ \citenamefont {Cocco}}]{Ri20}%
  \BibitemOpen
  \bibfield  {author} {\bibinfo {author} {\bibfnamefont {F.}~\bibnamefont
  {Rizzato}}, \bibinfo {author} {\bibfnamefont {A.}~\bibnamefont {Coucke}},
  \bibinfo {author} {\bibfnamefont {E.}~\bibnamefont {de~Leonardis}}, \bibinfo
  {author} {\bibfnamefont {J.~P.}\ \bibnamefont {Barton}}, \bibinfo {author}
  {\bibfnamefont {J.}~\bibnamefont {Tubiana}}, \bibinfo {author} {\bibfnamefont
  {R.}~\bibnamefont {Monasson}},\ and\ \bibinfo {author} {\bibfnamefont
  {S.}~\bibnamefont {Cocco}},\ }\bibfield  {title} {\bibinfo {title}
  {{Inference of compressed Potts graphical models}},\ }\href
  {https://doi.org/10.1103/PhysRevE.101.012309} {\bibfield  {journal} {\bibinfo
   {journal} {Phys. Rev. E}\ }\textbf {\bibinfo {volume} {101}},\ \bibinfo
  {pages} {012309} (\bibinfo {year} {2020})}\BibitemShut {NoStop}%
\bibitem [{\citenamefont {Stauffer}(2008)}]{St08}%
  \BibitemOpen
  \bibfield  {author} {\bibinfo {author} {\bibfnamefont {D.}~\bibnamefont
  {Stauffer}},\ }\bibfield  {title} {\bibinfo {title} {{Social applications of
  two-dimensional Ising models}},\ }  \href
  {https://doi.org/10.1119/1.2779882} {\bibfield  {journal} {\bibinfo  {journal}
  {Am. J. Phys.}\ }\textbf {\bibinfo {volume} {76}},\  \bibinfo
  {pages} {470-473} (\bibinfo {year}
  {2008})}\BibitemShut {NoStop}%
\bibitem [{\citenamefont {DeGiuli}(2019)}]{De19f}%
  \BibitemOpen
  \bibfield  {author} {\bibinfo {author} {\bibfnamefont {E.}~\bibnamefont
  {DeGiuli}},\ }\bibfield  {title} {\bibinfo {title} {{Random Language
  Model}},\ }\href {https://doi.org/10.1103/PhysRevLett.122.128301} {\bibfield
  {journal} {\bibinfo  {journal} {Phys. Rev. Lett.}\ }\textbf {\bibinfo
  {volume} {122}},\ \bibinfo {pages} {128301} (\bibinfo {year}
  {2019})}\BibitemShut {NoStop}%
\bibitem [{\citenamefont {Castellano}\ \emph {et~al.}(2009)\citenamefont
  {Castellano}, \citenamefont {Fortunato},\ and\ \citenamefont
  {Loretto}}]{Ca09}%
  \BibitemOpen
  \bibfield  {author} {\bibinfo {author} {\bibfnamefont {C.}~\bibnamefont
  {Castellano}}, \bibinfo {author} {\bibfnamefont {S.}~\bibnamefont
  {Fortunato}},\ and\ \bibinfo {author} {\bibfnamefont {V.}~\bibnamefont
  {Loretto}},\ }\bibfield  {title} {\bibinfo {title} {{Statistical physics of
  social dynamics}},\ }\href {https://doi.org/10.1103/RevModPhys.81.591}
  {\bibfield  {journal} {\bibinfo  {journal} {Rev. Mod. Phys.}\ }\textbf
  {\bibinfo {volume} {81}},\ \bibinfo {pages} {591-646} (\bibinfo {year}
  {2009})}\BibitemShut {NoStop}%
\bibitem [{\citenamefont {Jim{\'{e}}nez}\ \emph {et~al.}(2007)\citenamefont
  {Jim{\'{e}}nez}, \citenamefont {Tiampo},\ and\ \citenamefont
  {Posadas}}]{Ji07}%
  \BibitemOpen
  \bibfield  {author} {\bibinfo {author} {\bibfnamefont {A.}~\bibnamefont
  {Jim{\'{e}}nez}}, \bibinfo {author} {\bibfnamefont {K.~F.}\ \bibnamefont
  {Tiampo}},\ and\ \bibinfo {author} {\bibfnamefont {A.~M.}\ \bibnamefont
  {Posadas}},\ }\bibfield  {title} {\bibinfo {title} {{An Ising model for
  earthquake dynamics}},\ }\href {https://doi.org/10.5194/npg-14-5-2007}
  {\bibfield  {journal} {\bibinfo  {journal} {Nonlin. Processes Geophys.}\
  }\textbf {\bibinfo {volume} {14}},\ \bibinfo {pages} {5-15} (\bibinfo {year}
  {2007})}\BibitemShut {NoStop}%
\bibitem [{\citenamefont {Lee}\ \emph {et~al.}(2015)\citenamefont {Lee},
  \citenamefont {Broedersz},\ and\ \citenamefont {Bialek}}]{Le15c}%
  \BibitemOpen
  \bibfield  {author} {\bibinfo {author} {\bibfnamefont {E.~D.}\ \bibnamefont
  {Lee}}, \bibinfo {author} {\bibfnamefont {C.~P.}\ \bibnamefont {Broedersz}},\
  and\ \bibinfo {author} {\bibfnamefont {W.}~\bibnamefont {Bialek}},\
  }\bibfield  {title} {\bibinfo {title} {{Statistical Mechanics of the US
  Supreme Court}},\ }\href {https://doi.org/10.1007/s10955-015-1253-6}
  {\bibfield  {journal} {\bibinfo  {journal} {J. Stat. Phys.}\ }\textbf
  {\bibinfo {volume} {160}},\ \bibinfo {pages} {275-301} (\bibinfo {year}
  {2015})}\BibitemShut {NoStop}%
\bibitem [{\citenamefont {Barahona}(1982)}]{Bar82}%
  \BibitemOpen
  \bibfield  {author} {\bibinfo {author} {\bibfnamefont {F.}~\bibnamefont
  {Barahona}},\ }\bibfield  {title} {\bibinfo {title} {{On the computational
  complexity of Ising spin glass models}},\ }\href
  {https://doi.org/10.1088/0305-4470/15/10/028} {\bibfield  {journal} {\bibinfo
   {journal} {J. Phys. A}\ }\textbf {\bibinfo {volume} {15}},\ \bibinfo {pages}
  {3241-3253} (\bibinfo {year} {1982})}\BibitemShut {NoStop}%
\bibitem [{\citenamefont {Istrail}(2000)}]{Is00}%
  \BibitemOpen
  \bibfield  {author} {\bibinfo {author} {\bibfnamefont {S.}~\bibnamefont
  {Istrail}},\ }\bibfield  {title} {\bibinfo {title} {{Statistical Mechanics,
  three-dimensionality and NP-completeness: Universality of Intractability of
  the Partition Functions of the Ising Model Across Non-Planar Lattices}},\
  }in\ \href@noop {} {\emph {\bibinfo {booktitle} {STOC'00}}}\ (\bibinfo
  {publisher} {Portland, Oregon: ACM Press},\ \bibinfo {year} {2000})\ pp.\
  \bibinfo {pages} {87--96}\BibitemShut {NoStop}%
\bibitem [{Note1()}]{Note1}%
  \BibitemOpen
  \bibinfo {note} {P (NP) is the class of decision problems solvable in
  polynomial time by a (non)-deterministic Turing machine.}\BibitemShut {Stop}%
\bibitem [{\citenamefont {Moore}\ and\ \citenamefont {Mertens}(2011)}]{Mo11}%
  \BibitemOpen
  \bibfield  {author} {\bibinfo {author} {\bibfnamefont {C.}~\bibnamefont
  {Moore}}\ and\ \bibinfo {author} {\bibfnamefont {S.}~\bibnamefont
  {Mertens}},\ }\href
  {https://doi.org/10.1093/acprof:oso/9780199233212.001.0001} {\emph {\bibinfo
  {title} {{The nature of computation}}}}\ (\bibinfo  {publisher} {Oxford, UK: Oxford
  University Press},\ \bibinfo {year} {2011})\BibitemShut {NoStop}%
\bibitem [{\citenamefont {Mezard}\ and\ \citenamefont
  {Montanari}(2009)}]{Me09}%
  \BibitemOpen
  \bibfield  {author} {\bibinfo {author} {\bibfnamefont {M.}~\bibnamefont
  {Mezard}}\ and\ \bibinfo {author} {\bibfnamefont {A.}~\bibnamefont
  {Montanari}},\ }\href
  {https://doi.org/10.1093/acprof:oso/9780198570837.001.0001} {\emph {\bibinfo
  {title} {{Information, Physics, And Computation}}}}\ (\bibinfo  {publisher}
  {Oxford, UK: Oxford Graduate Texts},\ \bibinfo {year} {2009})\BibitemShut {NoStop}%
\bibitem [{\citenamefont {Lucas}(2014)}]{Lu14b}%
  \BibitemOpen
  \bibfield  {author} {\bibinfo {author} {\bibfnamefont {A.}~\bibnamefont
  {Lucas}},\ }\bibfield  {title} {\bibinfo {title} {{Ising formulations of many
  NP problems}},\ }\href {https://doi.org/10.3389/fphy.2014.00005} {\bibfield
  {journal} {\bibinfo  {journal} {Frontiers in Physics}\ }\textbf {\bibinfo
  {volume} {2}},\ \bibinfo {pages} {1} (\bibinfo {year} {2014})}\BibitemShut
  {NoStop}%
\bibitem [{\citenamefont {Stengele}\ \emph
  {et~al.}(2022{\natexlab{a}})\citenamefont {Stengele}, \citenamefont
  {Drexel},\ and\ \citenamefont {{De las Cuevas}}}]{St21}%
  \BibitemOpen
  \bibfield  {author} {\bibinfo {author} {\bibfnamefont {S.}~\bibnamefont
  {Stengele}}, \bibinfo {author} {\bibfnamefont {D.}~\bibnamefont {Drexel}},\
  and\ \bibinfo {author} {\bibfnamefont {G.}~\bibnamefont {{De las Cuevas}}},\
  }\bibfield  {title} {\bibinfo {title} {{Classical spin Hamiltonians are
  context-sensitive languages}},\ }\href
  {https://doi.org/10.1098/rspa.2022.0553} {\bibfield  {journal} {\bibinfo
   {journal} {Proc. R. Soc. A}\ }\textbf {\bibinfo {volume} {479}} (\bibinfo {year} {2023})}
   \BibitemShut {NoStop}%
   \bibitem [{\citenamefont {Chomsky}(1963)}]{Ch63}%
   \BibitemOpen
   \bibfield  {author} {\bibinfo {author} {\bibfnamefont {N.}~\bibnamefont
   {Chomsky}},\ }\bibfield  {title} {\bibinfo {title} {{Formal Properties of
   Language}},\ }in\ \href@noop {} {\emph {\bibinfo {booktitle} {Handbook of
   Mathematical Psychology}}},\ \bibinfo {editor} {edited by\ \bibinfo {editor}
   {\bibfnamefont {D.}~\bibnamefont {Luce}}}\ (\bibinfo  {publisher} {John Wiley
   {\&} Sons.},\ \bibinfo {year} {1963})\ \BibitemShut
   {NoStop}%
 \bibitem [{\citenamefont {Chomsky}(1965)}]{Ch65}%
   \BibitemOpen
   \bibfield  {author} {\bibinfo {author} {\bibfnamefont {N.}~\bibnamefont
   {Chomsky}},\ }\href {https://mitpress.mit.edu/books/aspects-theory-syntax}
   {\emph {\bibinfo {title} {{Aspects of the theory of syntax}}}}\ (\bibinfo
   {publisher} {MIT Press},\ \bibinfo {year} {1965})\BibitemShut {NoStop}%
\bibitem [{\citenamefont {Kozen}(1997)}]{Ko97}%
   \BibitemOpen
   \bibfield  {author} {\bibinfo {author} {\bibfnamefont {D.~C.}\ \bibnamefont
   {Kozen}},\ }\href {https://doi.org/10.1007/978-1-4612-1844-9} {\emph
   {\bibinfo {title} {{Automata and Computability}}}}\ (\bibinfo  {publisher}
   {Springer},\ \bibinfo {year} {1997})\BibitemShut {NoStop}%
\bibitem [{\citenamefont {{De las Cuevas}}(2013)}]{De13b}%
  \BibitemOpen
  \bibfield  {author} {\bibinfo {author} {\bibfnamefont {G.}~\bibnamefont {{De
  las Cuevas}}},\ }\bibfield  {title} {\bibinfo {title} {{A quantum information
  approach to statistical mechanics -- a tutorial}},\ }\href
  {https://doi.org/10.1088/0953-4075/46/24/243001} {\bibfield  {journal}
  {\bibinfo  {journal} {J. Phys. B}\ }\textbf {\bibinfo {volume} {46}},\
  \bibinfo {pages} {243001} (\bibinfo {year} {2013})}\BibitemShut {NoStop}%
\bibitem [{\citenamefont {Rozenberg}(1997)}]{Ro97}%
  \BibitemOpen
  \bibinfo {editor} {\bibfnamefont {G.}~\bibnamefont {Rozenberg}},\ ed.,\ \href
  {https://doi.org/10.1142/3303} {\emph {\bibinfo {title} {{Handbook of graph
  grammars and computing by graph transformations}}}}\ (\bibinfo  {publisher}
  {World Scientific},\ \bibinfo {year} {1997})\BibitemShut {NoStop}%
\bibitem [{\citenamefont {{De las Cuevas}}\ and\ \citenamefont
  {Cubitt}(2016)}]{De16b}%
  \BibitemOpen
  \bibfield  {author} {\bibinfo {author} {\bibfnamefont {G.}~\bibnamefont {{De
  las Cuevas}}}\ and\ \bibinfo {author} {\bibfnamefont {T.~S.}\ \bibnamefont
  {Cubitt}},\ }\bibfield  {title} {\bibinfo {title} {{Simple universal models
  capture all classical spin physics}},\ }\href
  {https://doi.org/10.1126/science.aab3326} {\bibfield  {journal} {\bibinfo
  {journal} {Science}\ }\textbf {\bibinfo {volume} {351}},\ \bibinfo {pages}
  {1180-1183} (\bibinfo {year} {2016})}\BibitemShut {NoStop}%
\bibitem [{\citenamefont {{De las Cuevas}}(2020)}]{De20d}%
  \BibitemOpen
  \bibfield  {author} {\bibinfo {author} {\bibfnamefont {G.}~\bibnamefont {{De
  las Cuevas}}},\ }\bibfield  {title} {\bibinfo {title} {{Universality
  everywhere implies undecidability everywhere}},\ }\href
  {https://fqxi.org/community/forum/topic/3529} {\bibfield  {journal} {\bibinfo
   {journal} {FQXi Essay}\ } (\bibinfo {year} {2020})}\BibitemShut {NoStop}%
\bibitem [{\citenamefont {Stengele}\ \emph
  {et~al.}(2022{\natexlab{b}})\citenamefont {Stengele}, \citenamefont
  {Reinhart}, \citenamefont {Gonda},\ and\ \citenamefont {{De las
  Cuevas}}}]{St22a}%
  \BibitemOpen
  \bibfield  {author} {\bibinfo {author} {\bibfnamefont {S.}~\bibnamefont
  {Stengele}}, \bibinfo {author} {\bibfnamefont {T.}~\bibnamefont {Reinhart}},
  \bibinfo {author} {\bibfnamefont {T.}~\bibnamefont {Gonda}},\ and\ \bibinfo
  {author} {\bibfnamefont {G.}~\bibnamefont {{De las Cuevas}}},\ }\bibfield
  {title} {\bibinfo {title} {{A framework for universality across
  disciplines}},\ }\href@noop {} {\bibfield  {journal} {\bibinfo  {journal} {In
  preparation}\ } (\bibinfo {year} {2022}{\natexlab{b}})}\BibitemShut {NoStop}%
\bibitem [{\citenamefont {Rozenberg}\ \emph {et~al.}(1997)\citenamefont
  {Rozenberg}, \citenamefont {Salomaa},\ and\ \citenamefont {Salomaa}}]{Ro971}%
  \BibitemOpen
  \bibfield  {author} {\bibinfo {author} {\bibfnamefont {G.}~\bibnamefont
  {Rozenberg}}, \bibinfo {author} {\bibfnamefont {G.~R.~A.}\ \bibnamefont
  {Salomaa}},\ and\ \bibinfo {author} {\bibfnamefont {A.}~\bibnamefont
  {Salomaa}},\ }\href {https://link.springer.com/book/10.1007/978-3-642-59136-5} {\emph
  {\bibinfo {title} {{Handbook of Formal Languages: Volume 1. Word, Language,
  Grammar}}}},\ Handbook of Formal Languages\ (\bibinfo  {publisher}
  {Berlin, Germany: Springer},\ \bibinfo {year} {1997})\BibitemShut {NoStop}%
\bibitem [{\citenamefont {Hopcroft}\ and\ \citenamefont {Ullman}(1979)}]{Ho79}%
  \BibitemOpen
  \bibfield  {author} {\bibinfo {author} {\bibfnamefont {J.~E.}\ \bibnamefont
  {Hopcroft}}\ and\ \bibinfo {author} {\bibfnamefont {J.~D.}\ \bibnamefont
  {Ullman}},\ }\href@noop {} {\emph {\bibinfo {title} {{Introduction to
  Automata Theory, Languages, and Computation}}}}\ (\bibinfo  {publisher}
  {Boston, MA: Addison-Wesley Publishing Company},\ \bibinfo {year} {1979})\BibitemShut
  {NoStop}%
\bibitem [{\citenamefont {Rothe}(2005)}]{Ro05}%
  \BibitemOpen
  \bibfield  {author} {\bibinfo {author} {\bibfnamefont {J.}~\bibnamefont
  {Rothe}},\ }\href {https://doi.org/10.1007/3-540-28520-2} {\emph {\bibinfo
  {title} {{Complexity Theory and Cryptology: An Introduction to
  Cryptocomplexity}}}}\ (\bibinfo  {publisher}
  {Berlin, Germany: Springer}, \ \bibinfo {year} {2005})\BibitemShut {NoStop}%
\end{thebibliography}
\end{document}